\documentclass[acmsmall]{acmart}
\usepackage{graphicx}
\usepackage{xcolor}
\usepackage[normalem]{ulem}
\usepackage{multirow}
\usepackage{tabularx}
\usepackage{macro}

\usepackage{amsmath}

\usepackage{amssymb}
\usepackage{amsthm}

\usepackage{makecell}
\usepackage{thmtools}
\usepackage{thm-restate}

\usepackage{algorithm}
\usepackage[noend]{algpseudocode}

\usepackage{comment}
\usepackage{wrapfig}
\usepackage{mdframed}
\usepackage{enumitem}
\usepackage{caption}
\usepackage{subcaption}
\usepackage{xspace}
\usepackage{tcolorbox}

\setcopyright{rightsretained}
\acmDOI{10.1145/3808279}
\acmYear{2026}
\acmJournal{PACMPL}
\acmVolume{10}
\acmNumber{PLDI}
\acmArticle{201}
\acmMonth{6}
\acmSubmissionID{pldi26main-p161-p}
\received{2025-11-07}
\received[accepted]{2026-04-03}



\begin{document}

\title{Choose, Don’t Label: Multiple-Choice Query Synthesis for Program Disambiguation}

\author{Celeste Barnaby}
\orcid{0000-0001-7688-6133}
\affiliation{%
  \institution{University of Texas at Austin}
  \city{}
  \country{USA}
}
\email{celestebarnaby@utexas.edu}

\author{Danny Ding}
\orcid{0009-0009-3437-8064}
\affiliation{%
  \institution{University of Texas at Austin}
  \city{}
  \country{USA}
}
\email{dingyy@utexas.edu}

\author{Osbert Bastani}
\orcid{0000-0001-9990-7566}
\affiliation{%
  \institution{University of Pennsylvania}
  \city{}
  \country{USA}
}
\email{obastani@seas.upenn.edu}

\author{Işıl Dillig}
\orcid{0000-0001-8006-1230}
\affiliation{%
  \institution{University of Texas at Austin}
  \city{}
  \country{USA}
}
\email{isil@cs.utexas.edu}

\begin{abstract}
High-level specifications of code are inherently ambiguous, and prior systems have explored interactive techniques to help users clarify their intent and resolve such ambiguities. However, most existing approaches elicit supervision through labeled examples, which are often error-prone and may fail to capture user intent. This paper introduces a new active learning paradigm for program disambiguation based on \emph{multiple-choice queries}. In this paradigm, the system presents a small set of  high-level behaviors as multiple-choice options, and the user simply selects the intended one. Technically, each answer option corresponds to a Hoare triple that characterizes a cluster of semantically similar candidate programs. This formulation enables formal reasoning about the informativeness and interpretability of queries, and supports systematic construction of optimal queries. Building on this insight, we develop a new active learning algorithm and implement it in a tool called \toolname{}, which automatically synthesizes informative multiple-choice queries for program disambiguation. We evaluate \toolname{} across four domains spanning both symbolic and neurosymbolic settings and show that it produces intuitive, easy-to-answer queries and achieves efficient convergence. Most importantly,  \toolname{} identifies the intended program \emph{more reliably} than existing methods, while maintaining competitive runtime performance.

\end{abstract}

\maketitle

\section{Introduction}

Recent advances in program synthesis and large language models have made it increasingly practical to generate code from high-level intent. Yet such intent is often underspecified or open to interpretation, giving rise to multiple plausible solutions that are {semantically distinct}. To address this ambiguity, prior work has explored active learning strategies that identify the desired program through targeted queries posed to the user~\cite{flashprog, samplesy, learnsy, forest, smartlabel}. These approaches, however, predominantly rely on users to label concrete inputs, which has two key limitations. First, manual input labeling can be cumbersome and error-prone in structure-rich domains, such as those involving large tables or nested data structures. Second, most existing methods assume a fixed set of inputs over which the synthesized program will be evaluated and only guarantee correctness over those.

\begin{wrapfigure}{r}{0.45\linewidth}
\vspace{-0.5em}
\begin{mdframed}[backgroundcolor=gray!10, linewidth=0.5pt, roundcorner=5pt]
\noindent\textbf{Query:} If the input table contains multiple rows for the same user, the program should return:
\begin{enumerate}[label=(\alph*), itemsep=0pt, leftmargin=*]
\item  one row per user
\item  one row per (user, date) pair
\item  one row per original record
\end{enumerate}
\end{mdframed}
\vspace{-1em}
\caption{Example multiple-choice query.}
\label{fig:query-ex}
\vspace{-0.5em}
\end{wrapfigure}
This paper presents a new approach to active learning for code generation, where users answer high-level multiple-choice questions instead of labeling specific inputs. 
As shown in Figure~\ref{fig:query-ex}, these questions let users choose among distinct behavioral patterns: for example, whether the program should return one row per user, one per user–date combination, or one per original record. By replacing tedious example annotation with intuitive questions about high-level behavior, this approach allows  users to express semantic intent more directly, while also guiding the synthesizer to reason over broader behaviors.
Furthermore, this approach offers stronger correctness guarantees than most prior methods~\cite{learnsy,samplesy,smartlabel} and yields more accurate user responses in practice.

At the heart of our approach is a method for partitioning the program space into a small number of semantic equivalence classes. Each class is defined by a shared precondition $\precond$ and a unique postcondition $\postcond_i$, yielding a structured query $
Q = (\precond, \postcond_1, \ldots, \postcond_k)
$. Here, $\precond$ describes a class of inputs (e.g., \emph{table containing multiple rows for the same user}), whereas each $\postcond_i$ specifies a possible property of the corresponding output (e.g., \emph{returning one row per user}). This structure naturally induces a multiple-choice question: the precondition sets up the scenario under which the program's behavior should be evaluated, and the postconditions become the candidate answers. Thus, answering the multiple-choice query reduces to selecting which of several Hoare-style specifications
$\{\precond\}\,P\,\{\postcond_i\}$ represents the user’s intent. 

The key challenge is selecting structured queries that are both \emph{informative} (i.e., best disambiguate programs) and \emph{interpretable} (i.e., easily understood by users). To address this challenge, our method first identifies a region of the input space where a maximal number of candidate programs exhibit different behavior. This region is summarized as a {precondition} that strikes a good balance between discriminative power and simplicity.  Once a precondition $\precond$ is fixed, our method clusters programs based on their semantic behavior under $\precond$, ensuring that the clusters are as even as possible so that \emph{any} user answer eliminates a large fraction of programs. Finally, each cluster's behavior is  summarized by a concise postcondition that separates each group from all of the others. This yields a complete and mutually exclusive partition of candidate programs and can be directly translated into a multiple-choice query posed in natural language.

We implemented our approach in a new tool called \toolname and evaluated it on four domains: table transformations, JSON transformations, batch image editing, and image search. In our experiments, \toolname achieved perfect accuracy across all tasks, whereas prior active learning methods failed on up to 36\% of benchmarks. Furthermore, it posed questions that users answered 38\% more accurately, and achieved comparable or better efficiency than existing approaches in terms of query selection efficiency, user response time, and interaction rounds.


To summarize, this paper makes the following key contributions:
\begin{itemize}[leftmargin=*]
\item {\bf New paradigm for interactive synthesis.}
We introduce a new interactive synthesis framework in which users answer multiple-choice questions about high-level program behavior, rather than labeling concrete inputs. Each question expresses a logical relation between inputs and outputs, allowing users to communicate intent at a semantic level (Section~\ref{sec:problem}).
\item {\bf Principled query synthesis algorithm.}
We design a query synthesis algorithm that formulates user queries as an optimization problem over \emph{informativeness} and \emph{interpretability}, solved through a combination of semantic reasoning, clustering, and interpolation (Section~\ref{sec:algo}).
\item {\bf End-to-end implementation and evaluation.}
We realize these ideas in a new tool called  \toolname (Section~\ref{sec:impl}), and demonstrate its effectiveness across four diverse domains. Our evaluation shows that \toolname achieves perfect accuracy and consistently improves user response correctness, while maintaining comparable or better efficiency compared to example-based approaches (Section~\ref{sec:eval}).
\end{itemize}




\section{Motivating Scenario \& Overview}

This section illustrates our approach through a simple example inspired by \ImageEye, a system for programmatic image manipulation from prior work ~\cite{imageeye}. \ImageEye  learns image editing programs from labeled examples, where each input is an original image and the output is its edited counterpart.

\begin{figure}[t]
  \centering
  \setlength{\tabcolsep}{2pt}

  \begin{subfigure}[t]{0.32\linewidth}
    \centering
    \includegraphics[width=\linewidth]{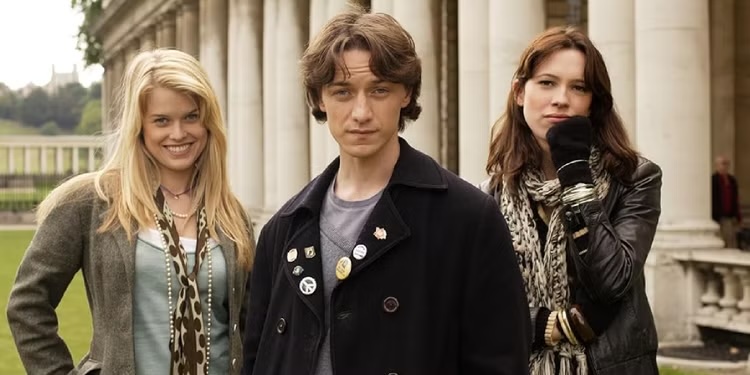}%
    \caption{Photo with Alice and  friends.}
    \label{fig:alice-io-input}
  \end{subfigure}\hfill
  %
  \begin{subfigure}[t]{0.32\linewidth}
    \centering
    \includegraphics[width=\linewidth]{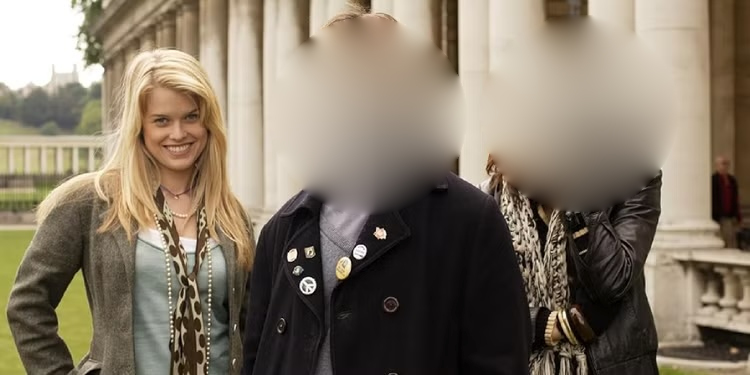}%
    \caption{Labeled output.}
    \label{fig:alice-io-output}
  \end{subfigure}\hfill
  %
  \begin{subfigure}[t]{0.32\linewidth}
    \centering
    \includegraphics[width=\linewidth]{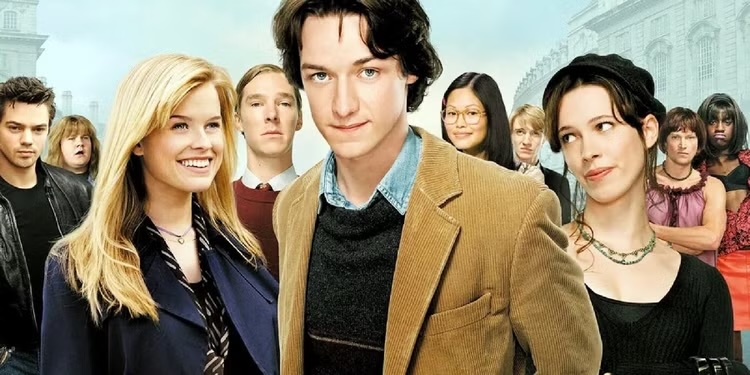}%
    \caption{Alice among many people.}
    \label{fig:alice-io-crowd}
  \end{subfigure}

  \vspace{-0.6em}
  \caption{Input–output demonstration (a,b) and a scenario (c) with high annotation burden.}
  \label{fig:alice-io}
\end{figure}

To illustrate the limitations of relying on labeled examples, consider a user, Alice, who wants to apply a privacy-preserving transformation to her photo collection: \emph{blur the faces of all people except herself}. In \ImageEye, this task can be achieved by a program that detects all faces, 
compares each against the target, and applies the \texttt{Blur} operator to non-matching faces (Program 1 in Figure~\ref{fig:alice-programs}). Now, suppose Alice provides the input–output example in Figures \ref{fig:alice-io-input}–\ref{fig:alice-io-output}, where her face remains clear while the others are blurred. Although this example captures her intent, it is ambiguous: several other programs, such as blurring all \emph{non-smiling} faces or those with \emph{non-blonde} hair, would produce the same output and therefore remain viable in the \emph{hypothesis space} (shown in Figure~\ref{fig:alice-programs}).


\paragraph{Existing methods.}
Prior approaches address such ambiguity in two main ways. The first is to present a candidate program to the user, who must inspect its source code or outputs and iteratively provide additional examples. This paradigm is both effort-intensive and prone to error. The second employs active learning to select the \emph{most informative} example for labeling~\cite{samplesy,learnsy,smartlabel}. Although this strategy makes the interaction more systematic, it can select examples that are difficult to annotate, such as images containing many people with diverse attributes (e.g., Figure~\ref{fig:alice-io-crowd}), increasing both annotation time and the likelihood of mistakes. Moreover, because these methods only reason over a fixed set of inputs, their correctness guarantees are limited to that dataset and do not ensure that the resulting program reflects the user’s intent on new data.

\begin{figure}[t]
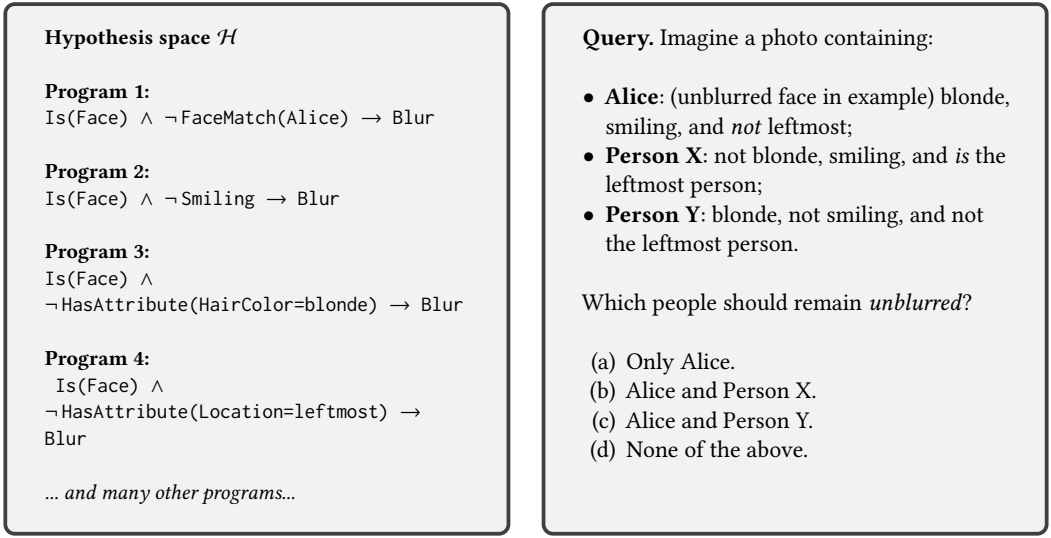

  \centering
  \begin{subfigure}[t]{0.485\linewidth}
    \begin{tcolorbox}[equal height group=H, valign=top]
      \raggedright
      \footnotesize\textbf{Hypothesis space $\mathcal{H}$}\par \ \\ 

      \textbf{Program 1:} \\ 
      {\ttfamily 
       Is(Face) $\wedge$ $\neg$\,FaceMatch(Alice) $\rightarrow$ Blur}  \\ \ \\ 
       \textbf{Program 2:} \\ 
           {\ttfamily  Is(Face) $\wedge$ $\neg$\,Smiling $\rightarrow$ Blur}  \\ \ \\ 
                  \textbf{Program 3:} \\ 
        {\ttfamily  Is(Face) $\wedge$ $\neg$\,HasAttribute(HairColor=blonde) $\rightarrow$ Blur}  \\ \ \\ 
               \textbf{Program 4:} \\ 
        {\ttfamily \ Is(Face) $\wedge$ $\neg$\,HasAttribute(Location=leftmost) $\rightarrow$ Blur} \\  \ \\ 
        \emph{... and many other programs...}
      \vspace{2pt}
    \end{tcolorbox}
    \caption{Hypothesis space.}
    \label{fig:alice-programs}
  \end{subfigure}\hfill
  \begin{subfigure}[t]{0.485\linewidth}
  \small
    \begin{tcolorbox}[equal height group=H, valign=top]
      \raggedright
      \textbf{Query.}  Imagine a photo containing: \ \\ \ \\ \par
      \begin{itemize}[leftmargin=*,nosep]
        \item \textbf{Alice}: (unblurred face in example) blonde, smiling, and \emph{not}  leftmost;
        \item \textbf{Person X}: not blonde, smiling, and \emph{is} the leftmost person;
        \item \textbf{Person Y}: blonde, not smiling, and not the leftmost person.
      \end{itemize}
      \ \\  \ \\ 
      Which people should remain \emph{unblurred}? \\ \ \\ 
      \begin{enumerate}[label=(\alph*),leftmargin=*,nosep]
        \item Only Alice.
        \item Alice and Person X.
        \item Alice and Person Y.
        \item None of the above.
      \end{enumerate}
    \end{tcolorbox}
    \caption{Multiple-choice query for running example}
    \label{fig:mc-query-alice}
  \end{subfigure}

  \vspace{-0.1in}
  \caption{ Programs consistent with user's example (a) and the generated multiple-choice query (b) }
  \label{fig:alice-side-by-side}
  \vspace{-0.15in}
\end{figure}

\paragraph{Our approach.}
Instead of asking users to label concrete examples, our method poses a multiple-choice question that highlights key differences among candidate programs. For the running example, Figure~\ref{fig:mc-query-alice} describes a hypothetical image that contains people with different attributes (i.e., hair color, position, and expression), and asks which faces should remain unblurred. Each answer reflects a distinct interpretation of the user’s intent and is designed to eliminate a corresponding subset of programs from the hypothesis space.

To generate such a query, \toolname constructs a \emph{text description} of a hypothetical image that makes candidate programs disagree. It first identifies \emph{distinguishing predicates}, or interpretable conditions on object attributes where two programs produce different results. For example, for the first two programs in Figure~\ref{fig:alice-programs}, a distinguishing predicate might state that ``Alice is not smiling, but someone else is.'' From this set of predicates $\Phi$, \toolname synthesizes a concise precondition $\phi$ that separates as many program pairs as possible while remaining easy to interpret. In our example, $\phi$ corresponds to the scene in Figure~\ref{fig:mc-query-alice}, which depicts Alice and two others with attributes chosen to maximally differentiate the remaining candidates.

Given the discriminating scene $\phi$, \toolname symbolically executes each program in $\mathcal{H}$ and groups those that produce the same result. Each group represents a distinct outcome under $\phi$ and serves as the basis for a multiple-choice answer. To keep the query concise, \toolname merges these groups into at most four balanced clusters so that any user selection eliminates a roughly equal share of the remaining candidates. Each cluster is then summarized by a \emph{separator}, a logical condition distinguishing it from the others, which is rendered in natural language (using an LLM) to form the answer choices, as shown in Figure~\ref{fig:mc-query-alice}.

\section{Problem Formulation}\label{sec:problem}

In this section, we formalize the problem addressed in the rest of this paper. 

\subsection{Query Space}
Given a finite hypothesis space $\mathcal{H}$ of candidate programs consistent with the user's initial specification, 
the goal is to identify the intended program $P^* \in \mathcal{H}$ through structured logical queries:
\[
\query = (\precond, \postcond_1, \dots, \postcond_k),
\]
where $\precond$ is a precondition describing a family of inputs, and $\postcond_i$ is a postcondition describing a distinct program behavior for that input family.
The pair $(\precond, \postcond_i)$ is called a \emph{scenario}. 
When the user selects option $i$, all programs inconsistent with $\postcond_i$ are eliminated, thereby refining $\mathcal{H}$.

For a query to be meaningful, its scenarios must partition the hypothesis space $\hs$ relative to 
$\precond$. That is, each program in $\hs$ should fall into exactly one scenario, ensuring that the 
user’s answer can be used to unambiguously refine $\hs$. We capture this requirement as follows.

\begin{definition}{\bf (Valid query).}~\label{def:valid-query}
Given hypothesis space $\hs$, a query 
$\query = (\precond, \postcond_1, \dots, \postcond_k)$ is \emph{valid} if
\[
\begin{array}{llll}
\textbf{(Mutual exclusion)} & \quad & 
\forall P \in \hs.\ \forall i \neq j.\ 
\neg \big( \models \{ \precond \}\, P \,\{ \postcond_i \} \ \wedge\ 
\models \{ \precond \}\, P \,\{ \postcond_j \} \big). \\
\textbf{(Coverage)} & \quad & 
\forall P\in\hs.\ \exists i.\ \models  \{ \precond \}\, P \,\{ \postcond_i \}.
\end{array}
\]
\end{definition}
Mutual exclusion ensures that no program in $\hs$ can satisfy two different scenarios 
with the same precondition, so the answer is unambiguous. Coverage ensures that every program in $\hs$ is consistent with at least one scenario, ensuring that one of the answers is correct. Now, given query $\query =  (\precond, \postcond_1, \dots, \postcond_k)$, for each $i$, define $\hs_{\precond,i} \;=\; \{\, P \in \hs \mid \{ \precond \}\, P \,\{ \postcond_i \}\,\}$
to be the programs in $\mathcal{H}$ that satisfy $(\precond, \postcond_i)$. By Definition~\ref{def:valid-query}, $\{ \hs_{\phi, i} \}_{i=1}^k$ is a complete partition of $\mathcal{H}$ (i.e., they are disjoint and cover $\mathcal{H}$). If the user selects answer $\postcond_i$, we can refine $\mathcal{H}$ to $\hs_{\precond,i}$.

\subsection{Query Selection Problem}\label{sec:orig-problem-stmt}

At a high level, the \emph{query selection problem} is to choose a query $\query$ that balances two goals. The first is \emph{disambiguation power}: each possible answer should eliminate a large portion of the hypothesis space. The second is \emph{interpretability}: the pre- and postconditions should form a question that users can easily understand and answer.

To formalize the problem, we assume two predicate universes: a precondition universe $\Upre$ and a postcondition universe $\Upost$. Each universe defines a finite set of atomic predicates that can be combined to form candidate pre- and postconditions, and we assume that these universes are literal-closed (i.e., $\forall a \in \mathcal{U}$, we also have $\neg a \in \mathcal{U}$). To keep queries interpretable, we restrict both pre- and postconditions to be \emph{cubes}, that is, conjunctions of atoms drawn from their respective universes.
Then, the query space is
\[
\mathcal{Q}=\{Q=(\precond, \postcond_1, \dots, \postcond_k)\mid \precond \in \textsc{Cube}(\Upre),\ \postcond_i \in \textsc{Cube}(\Upost),\ Q\text{ is valid}\},
\]
where validity is as in Definition~\ref{def:valid-query} and $\textsc{Cube}(\mathcal{U})$ is the space of cubes over predicate universe $\mathcal{U}$.

Next, \emph{disambiguation power} captures the efficacy of a query $Q$ at pruning the hypothesis space.
\begin{definition}[{\bf Disambiguation power}]
Let $\query = (\precond, \postcond_1, \dots, \postcond_k)$ be a query, and let $\hs_{\precond,i}$ denote the set of programs in $\hs$ that satisfy $\postcond_i$ under $\precond$. Then the disambiguation power of $\query$ is defined as:
\[
\dpower(\query) 
= \min_{i \in [1,k]} \frac{\sum_{j \ne i} |\hs_{\precond,j}|}{|\hs|}
= \left ( 1 - \max_{i \in [1,k]} \frac{|\hs_{\precond,i}|}{|\hs|} \right ) 
\]
\end{definition}
Intuitively, for any answer provided by the user, we can eliminate at least $\dpower(\query)$ fraction of programs from the hypothesis space. Thus, we want to select a query that maximizes disambiguation power. However, we also want to ensure the query is interpretable. We define the \emph{complexity} $\complexity(\query)$ of a query $\query$ based on the syntactic complexity of $\precond$ and the $\postcond_i$'s (e.g., the number of atoms in these logical formulas), reflecting how difficult the query is to interpret.
Now, we have:
\begin{definition}{\bf (Query selection problem)}\label{def:problem-orig}
Given $\Upre,\Upost$, $\hs$, and a hyperparameter $\lambda > 0$ balancing disambiguation power and complexity, the \emph{query selection problem} is to compute
\[
\query^* = \argmax_{\query \in \mathcal{Q}} \left[ \dpower(\query) - \lambda \cdot \complexity(\query) \right].
\]
\end{definition}

\subsection{Decomposed Problem Formulation}

Directly solving the query selection  problem from Section~\ref{sec:orig-problem-stmt}  is computationally challenging since it requires jointly optimizing over pre- and postconditions. This combined search space is too large for exhaustive evaluation to be feasible. Furthermore, optimizing interpretability at the level of the entire query makes it difficult to isolate and control the complexity of individual conditions. To address these challenges, we adopt a decomposed formulation. We first select a satisfiable precondition $\precond$ that captures semantically meaningful behavioral distinctions while remaining interpretable. Then, we construct postconditions $\postcond_i$ conditioned on $\precond$. To formalize this decomposed problem, we first define the \emph{disambiguation power} of a precondition $\precond$ in isolation as the number of (unordered) program pairs in the hypothesis space $\hs$ that it distinguishes:
\[
\dppre(\precond) = \left| \left\{ (P_1, P_2)  \mid P_1, P_2 \in \hs \wedge \forall x.\, \precond(x) \Rightarrow P_1(x) \ne P_2(x) \right\} \right|.
\]
\begin{definition}[{\bf Decomposed Query Selection Problem}]
\label{def:decomposed-query-selection}
Given $\Upre,\Upost$, $\hs$, and hyperparameters $\lambda_{\mathrm{pre}}$ and $\lambda_{\mathrm{post}}$, the \emph{Decomposed Query Selection Problem} is to compute
\begin{align}
\label{eq:preobjective}
\precond^*&\in\argmax_{\precond \in \textsc{Cube}(\Upre)}\left[
\dppre(\precond) - \lambda_{\mathrm{pre}} \cdot \complexity(\precond)
\right] \\
\label{eq:postobjective}
(\postcond_1^*, \dots, \postcond_k^*) &\in \argmax_{\substack{(\postcond_1, \dots, \postcond_k) \in \textsc{Cube}(\Upost)^k}}
\left[
\dpower(\precond^*, \postcond_1, \dots, \postcond_k) - \lambda_{\mathrm{post}} \cdot \sum_{i=1}^k \complexity(\postcond_i)
\right].
\end{align}
\end{definition}
We note that this decomposed formulation is not equivalent to the joint objective in Definition ~\ref{def:problem-orig}: fixing the precondition before optimizing postconditions may exclude queries that would score higher under simultaneous optimization. However, the decomposition offers two practical advantages that justify this trade-off. First, it makes the optimization tractable by replacing a single search over the combined space of all $(\precond, \postcond_1, \dots, \postcond_k)$ tuples with two smaller, sequential subproblems. Second, it allows interpretability to be controlled at the level of individual conditions, since the complexity of the precondition and postconditions can be penalized independently. Intuitively, a precondition $\precond$ that distinguishes many candidate programs naturally enables informative queries: when $\precond$ induces diverse program behaviors, it becomes easier to construct postconditions that divide the hypothesis space evenly, leading to high disambiguation power.




\section{Active Learning Algorithm}
\label{sec:algo}


\begin{algorithm}[t]
\caption{Main Active Learning Loop}
\label{alg:main}
\begin{algorithmic}[1]
\Require Initial hypothesis space $\hs$; predicate universes $\Upre$, $\Upost$
\Ensure A semantically unique program $P^* \in \hs$

\While{true} \label{alg:main:loop-start}
    \If{$\mathsf{NumUnique}(\hs) = 1$} \label{alg:main:check-termination}
         \Return $P^* \in \hs$ \label{alg:main:return}
    \EndIf

    \ForAll{distinct pairs $(P_i,P_j)\in \hs\times\hs$} \label{alg:main:pair-loop}
        \State $\Phi(i,j) \gets \Call{GetDistinguishing}{P_i,P_j,\Upre}$ \label{alg:main:compute-distinguishing}
    \EndFor \label{alg:main:end-pair-loop}

    \State $\precond \gets \Call{GetBestPrecondition}{\Phi}$ \label{alg:main:get-best-pre}

    \If{$\precond = \bot$} \label{alg:main:precondition-fail}
         $\Upre \gets \Call{RefinePredicates}{\hs,\Upre}$; \textbf{continue} \label{alg:main:refine-preconds}
    \EndIf

      \State $\query \gets \Call{GenerateQuery}{\precond,\hs,\Upost}$ \label{alg:main:generate-query} 

    \State $\postcond \gets \Call{QueryUser}{\query}$ \label{alg:main:query-user}
    \State $\hs \gets \{\,P\in\hs \mid \{\precond\}\,P\,\{\postcond\}\}$ \label{alg:main:refine-hypothesis}

\EndWhile \label{alg:main:loop-end}
\end{algorithmic}
\end{algorithm}

Algorithm~\ref{alg:main} presents our top-level procedure for solving the decomposed query selection problem. 
Given an initial hypothesis space $\hs$ and predicate universes $\Upre$ and $\Upost$, 
the algorithm iteratively interacts with the user until all remaining programs in $\hs$ are semantically equivalent. 
Each iteration begins by identifying preconditions under which pairs of programs $P_i, P_j \in \hs$ exhibit different behaviors 
(line~\ref{alg:main:compute-distinguishing}). 
Concretely, it computes a set of \emph{distinguishing predicates} $\Phi(i, j)$ over $\Upre$, 
where each $\varphi \in \Phi(i, j)$ is a sufficient condition for the outputs of $P_i$ and $P_j$ to differ -- 
that is, 
$\forall x.\ \varphi(x) \Rightarrow P_i(x) \neq P_j(x)$. This ensures that $\varphi$ captures an entire region where the two programs are semantically incompatible, so any input satisfying $\varphi$ witnesses their disagreement. 
To obtain these predicates, the algorithm first derives the weakest precondition under which $P_i$ and $P_j$ differ, 
then extracts its prime implicants~\cite{prime} over $\Upre$. 
This subroutine is summarized in Algorithm~\ref{alg:compute-distinguishing}.

Next, Algorithm~\ref{alg:main} invokes {\sc GetBestPrecondition} to compute a symbolic precondition $\precond$ 
that maximizes the objective in Eq.~(\ref{eq:preobjective}) from Definition~\ref{def:decomposed-query-selection}. 
If no predicate in $\Upre$ satisfies this objective, {\sc GetBestPrecondition} returns $\bot$, 
triggering {\sc RefinePredicates} (Algorithm~\ref{alg:refine-predicates}) to extend $\Upre$ with new atoms 
and restart the iteration (line~\ref{alg:main:refine-preconds}). 
Once a valid precondition $\precond$ is obtained, the algorithm calls {\sc GenerateQuery} to synthesize postconditions 
that jointly optimize Eq.~(\ref{eq:postobjective}) in Definition~\ref{def:decomposed-query-selection}. 
The resulting query $Q = (\precond, \postcond_1, \ldots, \postcond_k)$ is then translated into natural language using an LLM 
and presented to the user (line~\ref{alg:main:query-user}). 
If the user selects answer $j$, their intended program $P$ must satisfy the Hoare triple 
$\{ \precond \} P \{ \postcond_j \}$, 
and all programs violating this triple are removed from $\hs$ (line~\ref{alg:main:refine-hypothesis}). 
The process repeats until all remaining programs in $\hs$ are semantically equivalent 
(line~\ref{alg:main:check-termination}). 
The following subsections describe precondition synthesis and query generation in more detail.

{
Algorithm~\ref{alg:main} satisfies two key guarantees. The first says that the query selected at each step solves the Decomposed Query Selection Problem (Definition~\ref{def:decomposed-query-selection}); the second says that, assuming the user responds correctly, the algorithm converges to the correct program.\footnote{Proofs of all theorems are provided in Appendix ~\ref{sec:proofs}.}

\begin{restatable}{theorem}{Optimality}\label{thm:optimality}
\textsc{(Optimality of Query Selection)}. At each iteration, the query $(\precond,\postcond_1,...,\postcond_k)$ generated by Algorithm~\ref{alg:main} solves the Decomposed Query Selection Problem in Definition~\ref{def:decomposed-query-selection}.
\end{restatable}}

\begin{restatable}{theorem}{Main}\label{thm:main}
\textsc{(Correctness of Active Learning)}. Given a finite hypothesis space $\hs$ containing the ground-truth program, Algorithm~\ref{alg:main} always terminates and returns a program $P^* \in \hs$ that is semantically equivalent to the user's intended program (under the assumption that the user provides correct answers to each query).
\end{restatable}


\begin{figure}[t]
\centering
\begin{minipage}{0.48\textwidth}
\small
\begin{algorithm}[H]
\caption{\textsc{GetDistinguishing}$(P_1, P_2, \Upre)$}

\label{alg:compute-distinguishing}
\begin{algorithmic}[1]
\Require Programs $P_1, P_2$; universe $\Upre$
\Ensure Set $\mathcal{C}$ of distinguishing cubes over $\Upre$
\State $\varphi \gets \mathsf{WeakestPre}(  \mathsf{assert}(P_1(x) \neq P_2(x))$
\State $D \gets \textsc{DNF}(\varphi)$; $\mathcal{C} \gets \varnothing$
\ForAll{clause $d \in D$}
  \State $\mathcal{C} \gets \mathcal{C} \cup \Call{FindImplyingCubes}{d, \Upre}$
\EndFor
\State \Return $\mathcal{C}$
\end{algorithmic}
\end{algorithm}
\end{minipage}
\hfill
\begin{minipage}{0.48\textwidth}
\begin{algorithm}[H]
\caption{\textsc{RefinePredicates}$(\hs, \Upre)$}
\label{alg:refine-predicates}
\begin{algorithmic}[1]
\small
\Require Hypothesis space $\hs$;  universe $\Upre$
\Ensure Expanded universe $\Upre$
\ForAll{distinct pairs $(P_1,P_2)\in \hs\times\hs$}
\State $\varphi \gets \mathsf{WeakestPre}(  \mathsf{assert}(P_1(x) \neq P_2(x))$
  \ForAll{$\alpha \in \mathsf{Atoms}(\varphi)$}
    \If{\Call{Admissible}{$\alpha$}} $\Upre \gets \Upre \cup \{\alpha\}$
    \EndIf
  \EndFor
\EndFor
\State \Return $\Upre$
\end{algorithmic}
\end{algorithm}
\end{minipage}
\vspace{-0.1in}
\caption{Auxiliary procedures for finding distinguishing predicates and refining predicate universe. The  \textsc{Admissible} procedure subjects each atom to an application-specific admissibility test and may exclude predicates based on syntactic or semantic complexity. }
\vspace{-0.1in}
\end{figure}

\subsection{Precondition Synthesis}
\label{sec:precondition-synthesis}

\begin{figure}[t]
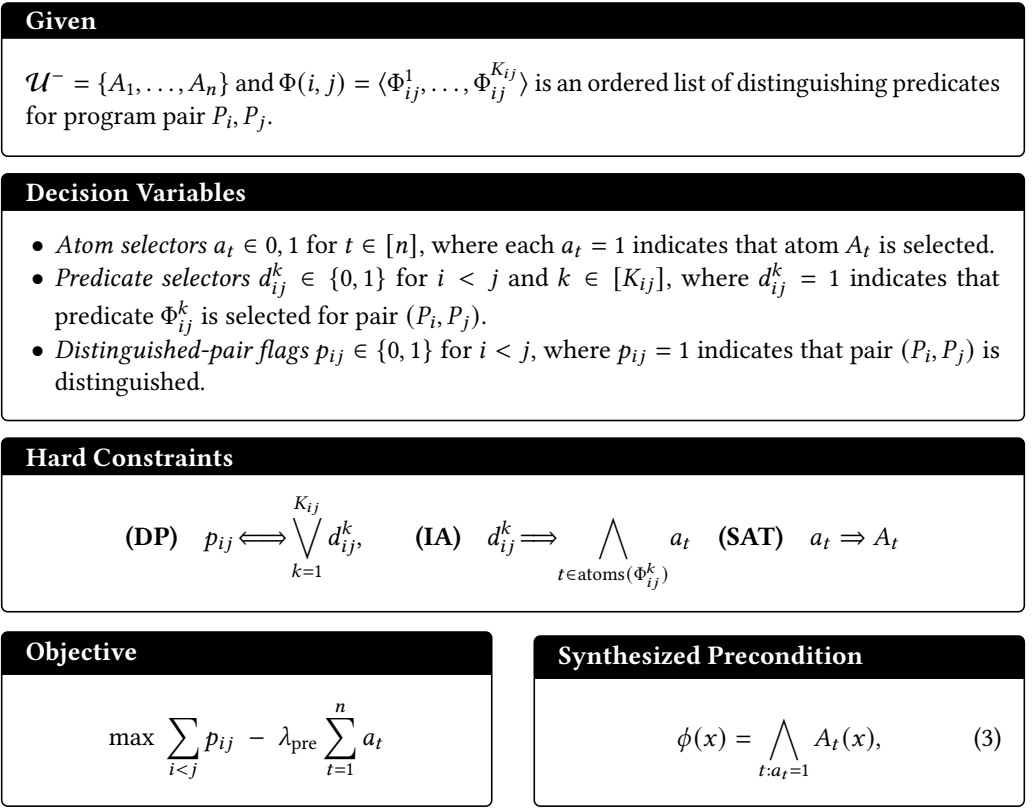

\centering
\begin{minipage}{0.98\linewidth}
\newtcolorbox{blk}[1]{colback=white,colframe=black,boxrule=0.6pt,arc=2pt,
  left=6pt,right=6pt,top=6pt,bottom=6pt,fonttitle=\bfseries,title={#1}}

\begin{blk}{Given}
$\Upre=\{A_1,\dots,A_n\}$ and  $\Phi(i,j)=\langle \Phi_{ij}^1,\dots,\Phi_{ij}^{K_{ij}}\rangle$
is an ordered list of distinguishing predicates for program pair $P_i, P_j$.
\end{blk}

\begin{blk}{Decision Variables}
\begin{itemize}[leftmargin=1.1em]
  \item \emph{Atom selectors} $a_t \in {0,1}$ for $t \in [n]$, where each $a_t = 1$ indicates that atom $A_t$ is selected.
\item \emph{Predicate selectors} $d_{ij}^k \in \{0,1\}$ for $i < j$ and $k \in [K_{ij}]$, 
where $d_{ij}^k = 1$ indicates that predicate $\Phi_{ij}^k$ is selected for pair $(P_i, P_j)$.  

\item \emph{Distinguished-pair flags} $p_{ij} \in \{0,1\}$ for $i < j$, 
where $p_{ij} = 1$ indicates that pair $(P_i, P_j)$ is distinguished.  
\end{itemize}
\end{blk}

\begin{blk}{Hard Constraints}
\vspace{-3pt}
\[
\begin{aligned}
\textbf{(DP)} \quad  p_{ij}\!\Longleftrightarrow\!\bigvee_{k=1}^{K_{ij}} d_{ij}^k, \quad 
&&\textbf{(IA)} \quad d_{ij}^k\!\Longrightarrow\!\!\bigwedge_{t\in\mathrm{atoms}(\Phi_{ij}^k)}\!a_{t} \quad  
\textbf{(SAT)}  \quad a_t \Rightarrow A_t
\end{aligned}
\]
\vspace{-5pt}
\end{blk}

\begin{minipage}[t]{0.48\linewidth}
\begin{blk}{Objective}
\[
\max\ \sum_{i<j} p_{ij}\;-\;\lambda_{\mathrm{pre}}\sum_{t=1}^{n} a_t
\]
\end{blk}
\end{minipage}\hfill
\begin{minipage}[t]{0.48\linewidth}
\begin{blk}{Synthesized Precondition}
\begin{align}
\label{eq:getbestprecond}
\precond(x) = \bigwedge_{t : a_t=1} A_t(x),
\end{align}
\end{blk}
\end{minipage}

\end{minipage}
\caption{Optimization modulo theory formulation of {\sc GetBestPrecondition}}
\label{fig:omt-encoding}
\vspace{-0.1in}
\end{figure}

We now explain the {\sc GetBestPrecondition} procedure for finding a single
predicate $\precond$ that optimizes the first objective (Eq.~\ref{eq:preobjective}) in
Definition~\ref{def:decomposed-query-selection}. Recall that our goal is to find
a conjunction of predicates over $\Upre$ that will differentiate as many
program pairs as possible, but with a complexity penalty. To find such a
precondition, our method utilizes the distinguishing predicates $\Phi$
precomputed for each program pair in Algorithm~\ref{alg:main}. Specifically,
rather than performing a blind search over all cubes in $\Upre$, we leverage
$\Phi$ to restrict attention to conjunctions that are known to be useful. Each distinguishing predicate in $\Phi(i,j)$ guarantees that $P_i$ and $P_j$ behave differently on this input region and thereby specifies exactly which atoms of $\Upre$ must be included for the precondition to distinguish this pair.
 Because each predicate in $\Phi(i,j)$ is a sufficient condition for $P_i$ and $P_j$ to differ, the semantic reasoning is already handled in Algorithm~\ref{alg:main}. The remaining task, addressed in {\sc GetBestPrecondition}, is to find a consistent subset of these witnesses and combine their atoms into a single global precondition.


We encode {\sc GetBestPrecondition} as an Optimization
Modulo Theory (OMT) problem in
Figure~\ref{fig:omt-encoding}. Our formulation introduces three families of
variables: atom selectors $a_t$, which determine which atomic predicates from
$\Upre$ are included in the final conjunction; predicate selectors $d_{ij}^k$,
which track which distinguishing predicate is used as a witness for program pair $(P_i, P_j)$;
and pair flags $p_{ij}$, which indicate whether the pair $(P_i, P_j)$ is successfully
distinguished. Constraints (DP)--(SAT) tie these variables together. First, (DP)
ensures that a pair is marked as distinguished iff at least one distinguishing predicate for that pair is selected.
Next, (IA) enforces that if a predicate $\Phi_{ij}^k$ is chosen (meaning $d_{ij}^k$ is assigned to true), then all atoms that
occur in it are also selected.
Finally, (SAT) ties the atom-selection variables to their corresponding predicates by asserting 
$a_t \Rightarrow A_t$ for every atom $A_t \in \Upre$. 
These implications define the synthesized precondition
$ 
\phi(x) = \bigwedge_{t : a_t = 1} A_t(x)
$, 
with the  hard constraints collectively ensuring that $\phi$ is always satisfiable. 


The objective function (shown under Objective in Figure~\ref{fig:omt-encoding}) trades off disambiguation power and complexity: the first term (sum over all $p_{ij}$'s)
rewards maximizing coverage by distinguishing as many pairs as possible, while
the second term (sum over $a_t$'s) penalizes the number of atoms selected to bias the solution toward
simpler and more interpretable preconditions. After solving  this OMT instance, we obtain  the precondition
$\phi$ by conjoining exactly those atoms $A_t$ whose corresponding indicator $a_t$ is true.


\begin{restatable}[Correctness of \textsc{GetBestPrecondition}]{theorem}{GetBestPrecond}
\label{cor:omt-equivalence}
\textsc{GetBestPrecondition} solves Eq.~(\ref{eq:preobjective}) in Definition~\ref{def:decomposed-query-selection}; i.e., letting $\precond(x)$ be as in Eq.~(\ref{eq:getbestprecond}) and $\complexity(\psi)$ be the number of conjuncts in $\psi$, then
\[
\precond \in \argmax_{\psi \in \textsc{Cube}(\Upre)}
\bigl[\dppre(\psi) - \lambda_{\mathrm{pre}}\cdot \complexity(\psi)\bigr].
\]
\end{restatable}

\begin{example} Consider the following hypothesis space of \textsc{ImageEye} programs:
{\small
\begin{align*}
\prog_1 :=& \texttt{ Is(Face)} \ \wedge \ \texttt{HasAttribute(HairColor=brown)} \rightarrow \texttt{Brighten} \\
\prog_2 :=& \texttt{ Is(Face)} \ \wedge \ \texttt{Smiling} \rightarrow \texttt{Brighten} \\ 
\prog_3 :=& \texttt{ Find(Is(Face), Is(Guitar), Above)} \rightarrow \texttt{Brighten}
\end{align*}}%
Here, $\prog_1$ brightens all faces with brown hair, $\prog_2$ brightens all smiling faces, and $\prog_3$ brightens all faces that are above guitars. Assuming the input image contains two objects $\{x_1, x_2\}$, suppose that \textsc{GetDistinguishing} generates the following constraints for each pair of programs:

{\footnotesize
\setlength{\jot}{2pt} 
\begin{align*}
\Phi_{1,2} =&\ \big\{
\redcolor{\texttt{HasLabel($x_1$, face)}} \redcolor{ $\wedge$ }
\redcolor{\texttt{HasHairColor($x_1$, brown)}} \redcolor{ $\wedge$ }
\redcolor{$\neg$\texttt{HasExpression($x_1$, smiling)}}, \\
&\quad
\text{\texttt{HasLabel($x_1$, face)}} \wedge
\text{\texttt{HasHairColor($x_1$, blonde)}} \wedge
\text{\texttt{HasExpression($x_1$, smiling)}}, 
\big\} \\[0.3em]
\Phi_{1,3} =&\ \big\{
\redcolor{\texttt{HasLabel($x_1$, face)}} \redcolor{ $\wedge$ }
\redcolor{\texttt{HasHairColor($x_1$, brown)}} \redcolor{ $\wedge$ }
\redcolor{$\neg$\texttt{HasLabel($x_2$, guitar)}}, \\
&\quad
\bluecolor{\texttt{HasLabel($x_1$, face)}} \bluecolor{ $\wedge$ }
\bluecolor{\texttt{HasHairColor($x_1$, blonde)}} \bluecolor{ $\wedge$ }
\bluecolor{\texttt{HasLabel($x_2$, guitar)}}  \bluecolor{ $\wedge$ }
\bluecolor{\texttt{Above($x_1$, $x_2$)}}
\big\} \\[0.3em]
\Phi_{2,3} =&\ \big\{
\text{\texttt{HasLabel($x_1$, face)}} \wedge
\text{\texttt{HasExpression($x_1$, smiling)}} \wedge
\neg\text{\texttt{HasLabel($x_2$, guitar)}}, \\
&\quad
\bluecolor{\texttt{HasLabel($x_1$, face)}} \bluecolor{ $\wedge$ }
\bluecolor{$\neg$}\bluecolor{\texttt{HasExpression($x_1$, smiling)}} \bluecolor{ $\wedge$ }
\bluecolor{\texttt{HasLabel($x_2$, guitar)}} \bluecolor{ $\wedge$ }
\bluecolor{\texttt{Above($x_1$, $x_2$)}}
\big\}
\end{align*}}%

Here, choosing the \redcolor{\text{red}} or the \bluecolor{\text{blue}} constraints maximizes distinguished program pairs. Since the red constraints contain fewer unique atoms, they will be selected to generate the precondition:

{\small\[
\begin{array}{c}
\text{\texttt{HasLabel($x_1$, face)}} \ \wedge \ \text{\texttt{HasHairColor($x_1$, brown)}} \ \wedge \  \\
\neg\text{\texttt{HasExpression($x_1$, smiling)}} \ \wedge \ \neg\text{\texttt{HasLabel($x_2$, guitar)}}
\end{array}
\]}%

This precondition means that the image contains two objects: the first is a face with brown hair that is not smiling, and the second is an object that is not a guitar. 

\end{example}

\subsection{Overview of Multiple-Choice Answer Generation}\label{sec:query-gen}

Now that we have computed the optimal precondition $\precond$, we need to generate optimal answer choices for $\precond$. At a high level, the precondition $\precond$ identifies the region of the input space on which the programs in $\hs$ exhibit qualitatively different behaviors, and the answer choices should group these behaviors into a small set of mutually exclusive, collectively exhaustive scenarios that the user can reliably distinguish.  Algorithm~\ref{alg:generate-query} summarizes how we generate these answer choices. 

To start with, Algorithm~\ref{alg:generate-query}  calls \textsc{GroupBySP} to partition the programs in $\hs$ based on their output behavior under $\precond$ (line~2), producing equivalence classes
\[
C_i = \{ P \in \hs \mid \mathsf{StrongestPost}(P, \precond) = \gamma_i \}
\qquad(\forall i\in\{1,...,N\}).
\]
Here, $\mathsf{StrongestPost}$ denotes the strongest postcondition; thus, each $C_i$ consists of all programs in $\hs$ whose output behaviors are indistinguishable under $\precond$.  In principle,  we could now directly translate each equivalence class $C_i$ into an answer choice based on $\gamma_i$. However, this approach is undesirable for two reasons: (1) the number of distinct clusters $N$ may be too large, resulting in an impractically long list of answer choices, and (2) the postconditions $\gamma_i$ are typically more complex than needed.



\begin{algorithm}[t]
\caption{\textsc{GenerateQuery}}
\label{alg:generate-query}
\begin{algorithmic}[1]
\Require Precondition $\precond$, hypothesis space $\hs$, postcondition universe $\Upost$
\Ensure Query $(\precond, \psi_1,\dots,\psi_k)$

\While{true}
  \State $\{C_1,\dots,C_N\} \gets \Call{GroupBySP}{\hs,\precond}$
  \State $(B_1,\dots,B_k) \gets \Call{MergeClusters}{\{C_i\},\,k, \precond}$
  \For{$i = 1,\dots,k$}
    \State $\psi_i \gets \Call{ConstructSeparator}{B_i,\{B_j\}_{j\neq i},\Upost}$
  \EndFor
  \If{$\exists\, i \in \{1,\dots,k\} \text{ with } \psi_i = \bot$}
    \State $\Upost \gets \Upost \cup \Call{RefinePostPredicates}{\hs,\precond,\Upost}$
  \Else \  \Return $(\precond,\psi_1,\dots,\psi_k)$
  \EndIf
\EndWhile

\AlgSeparator 

\Procedure{RefinePostPredicates}{$\hs,\precond,\Upost$}
  \State \Return $ \bigcup_{P\in \hs}\{\,a \in \mathsf{Atoms}(\mathsf{StrongestPost}(P,\precond)) \mid \mathsf{Admissible}(a)\,\}$
\EndProcedure
\end{algorithmic}
\end{algorithm}
To address these challenges, Algorithm~\ref{alg:generate-query} then invokes \textsc{MergeClusters} to coarsen the initial fine-grained partition $\{C_1, \dots, C_N\}$ into at most $k$ disjoint bins $B_1, \dots, B_k$, where $k$ is a small constant (typically 3 or 4). 
Each bin $B_i$ corresponds to answer choice $i$ in the final multiple-choice query, such that selecting option $i$ retains only the programs in $B_i$ and eliminates all others. 
\textsc{MergeClusters} seeks to produce bins of roughly equal size so that each answer removes a comparable portion of the hypothesis space. Next,
Algorithm~\ref{alg:generate-query} calls {\sc ConstructSeparator} (line 5) to convert each bin $B_i$ into a symbolic postcondition $\postcond_i$ over the predicate universe $\Upost$, ensuring that $\postcond_i$ is consistent with the behaviors in $B_i$ and excludes all programs in other bins. 
If no such postconditions can be synthesized (meaning {\sc ConstructSeparator} returns $\bot$ for at least one cluster), the algorithm calls \textsc{RefinePostPredicates} (defined in lines~9--10) to extend $\Upost$ with new atomic predicates derived from the strongest postcondition.

We next describe  the \textsc{MergeClusters} and \textsc{ConstructSeparator} procedures in more detail. We start with 
\textsc{ConstructSeparator} because it is internally used by \textsc{MergeClusters} to
evaluate the cost of candidate merges.

\begin{algorithm}[t]
\caption{\textsc{ConstructSeparator}$(B,\{B_1, \ldots, B_k\},\mathcal{U}^+)$}
\label{alg:sep}
\begin{algorithmic}[1]

\State $\Phi^+ \gets \bigvee_{P\in B}\mathsf{StrongestPost}(P,\precond)$ \Comment{Positive spec for target bin}
\State $\varphi_j \gets  \, \bigvee_{P\in B_j}\mathsf{StrongestPost}(P,\precond)$  \Comment{Negative specs for each $B_j$}
\If{$\exists j.\ \mathrm{SAT}(\Phi^+\land\varphi_j)$} \Return 
$\bot$ 
\EndIf
\State $\mathcal{A}\gets\{\,a\in\mathcal{U}^+\mid \mathrm{UNSAT}(\Phi^+\land\neg a)\,\}$ \Comment{All atoms implied by $\Phi^+$}
\State $S \gets \{ s_a \ | \ a \in \mathcal{A} \}$  \Comment{Indicators $s_a$ denoting that atom $a$ is chosen}
\State  $\mathcal{C}\gets\emptyset$ \Comment{Counterexamples}
\While{true}
\State $\psi \gets \textsc{MaxSAT}\big(\{\,\lnot s_a \mid a \in \mathcal{A}\,\},\, \mathcal{C}\big)$%
\Comment{Returns $\psi \triangleq \bigwedge_{s_a = 1} a$ with as few atoms as possible}

  \If{$\psi=\bot$} \Return $\bot$ 
  \EndIf
  \If{$\forall j. \ \mathrm{UNSAT}(\varphi_j\land\psi)$} \Return $\psi$ \EndIf
 \State Choose $j,m$ with $m\models(\varphi_j\land\psi)$
  \State $\mathcal{C}\gets \mathcal{C}\ \cup\ \Big\{\ \displaystyle\bigvee_{a\in\mathcal{A},\ m\models \neg a} s_a\ \Big\}$
\EndWhile
\end{algorithmic}
\end{algorithm}

\subsection{Separator Construction}\label{sec:construct-separator}

The goal of {\sc ConstructSeparator} (summarized in Algorithm~\ref{alg:sep}) 
is to synthesize a postcondition 
that captures the behavior of programs in the target bin while excluding all others. Specifically, it takes as input the target bin $B$, the negative bins $\{B_1,\ldots,B_k\}$, and the atom universe $\mathcal{U}^+$, and aims to find a postcondition $\psi$ that (1) holds for all behaviors represented by $B$, 
(2) rules out behaviors from every other bin $B_i$, (3) is a cube $\psi \in \mathsf{Cubes}(\mathcal{U}^+)$, and (4) contains as few atoms as possible. The first two conditions ensure that $\postcond$ is correct (i.e., it separates $B$ from the other bins), and the last two conditions ensure that it is interpretable (i.e., the answer choices are understandable).

Algorithm~\ref{alg:sep} starts  by constructing the \emph{positive specification} 
$\Phi^+ = \bigvee_{P\in B}\mathsf{StrongestPost}(P,\precond)$  for the target bin $B$
and, a \emph{negative specification} 
$\varphi_j = \bigvee_{P\in B_j}\mathsf{StrongestPost}(P,\precond)$ for every other $B_j$. 
If $\Phi^+\land\varphi_j$ is satisfiable for any $j$, then the behaviors of $B$ and $B_j$ overlap on at least one input. In this case, no separator $\psi$  can distinguish $B$ and $B_j$ because $\psi$ needs to hold for \emph{all} $B$ behaviors but reject \emph{all} $B_j$ behaviors.   Thus, in this case, the procedure terminates with $\bot$ to indicate failure.

Otherwise, the algorithm proceeds to construct a separator $\postcond$ that satisfies the four conditions mentioned earlier. Here, satisfying conditions (1) and (2) is straightforward: since the check on line 3 guarantees that $\Phi^+ \wedge (\bigvee_j \varphi_j)$ is unsatisfiable, the entailment $\Phi^+ \Rightarrow \neg(\bigvee_j \varphi_j)$ holds, and we can take $\psi$ to be a Craig interpolant for this entailment -- that is, $\Phi^+ \Rightarrow \psi$ and $\psi \Rightarrow \neg(\bigvee_j \varphi_j)$. Then, $\psi$ satisfies (1) since $\Phi^+ \Rightarrow \psi$ by definition, and (2) since $\psi \Rightarrow (\bigvee_j \varphi_j)$ implies that $\psi \wedge \varphi_j$ is unsatisfiable for every $j$. In other words, any Craig interpolant satisfies the first two conditions, but we must ensure that the interpolant is also a cube and contains as few atoms as possible to satisfy conditions (3) and (4).





\begin{wrapfigure}{r}{0.45\linewidth}
\vspace{-0.1in}
\begin{mdframed}[backgroundcolor=gray!10, linewidth=0.5pt, roundcorner=5pt]
\noindent\textbf{Query:} Suppose you have an input table with 2 rows and 2 columns. If the value in cell $(1,1)$ is $-1$ and the value in cell $(2, 1)$ is $0$, which of the following is true of your output table?:
\begin{enumerate}[label=(\alph*), itemsep=0pt, leftmargin=*]
\item  The table has 2 rows
\item  The table has 1 row, and the value in cell $(1,1)$ is $0$. 
\item  The table has 1 row, and the value in cell $(1, 1)$ is $-1$.
\end{enumerate}
\end{mdframed}
\vspace{-1em}
\caption{Generated query.}
\label{fig:mc-ex}
\vspace{-0.2in}
\end{wrapfigure}

Hence, rather than using an off-the-shelf interpolation tool, our method  
searches for a \emph{minimum cube interpolant} using a custom algorithm based on
counterexample-guided inductive synthesis (CEGIS). 
Specifically, it  treats each atom in $\mathcal{U}^+$ as a candidate building block, and incrementally constructs the interpolant by alternating between an \emph{optimization} phase and a \emph{verification} step. Given a set of counterexamples $\mathcal{C}$, the algorithm first attempts to solve an optimization problem subject to $\mathcal{C}$ and then checks whether the resulting solution is indeed a valid interpolant. If not, it strengthens the specification $\mathcal{C}$ and solves a more constrained optimization problem. Because each step preserves optimality subject to an over-approximation of the true specification, the first solution found is guaranteed to be the optimal interpolant.

In more detail, line 3 of Algorithm~\ref{alg:sep}  first restricts the search space to atoms $\mathcal{A}=\{\,a\in\mathcal{U}^+\mid \mathrm{UNSAT}(\Phi^+\land\neg a)\,\}$ that are already entailed by the positive
specification.  This is valid because any cube interpolant must consist only of atoms implied by $\Phi^+$. Then the CEGIS loop (lines 7--12) alternates between using \textsc{MaxSAT} to compute the optimal solution consistent with $\mathcal{C}$ (line 8) and searching for counterexamples (line 10). The \textsc{MaxSAT} problem is over the Boolean variables $s_a$ defined on line 5, where $s_a$ indicates whether atom $a\in\mathcal{A}$ is in $\postcond$ (i.e.,  $\psi$ is defined as $\bigwedge_{s_a = 1} a$). Then, the optimization problem  is to minimize the number of atoms in $\psi$ subject to $\mathcal{C}$ (line 6).

If the computed $\psi$ successfully separates  the positive and negative specifications (i.e., 
$\psi\land\varphi_j$ is unsatisfiable for all $j$), then it is returned as the solution (line 10). Otherwise, the solver produces a model
$m\models(\varphi_j\land\psi)$ for some $j$, from which a new blocking clause
$\bigvee_{a\in\mathcal{A},\,m\models\neg a}s_a$ is derived.
This clause enforces that any  solution must include at least one atom
that contradicts the counterexample $m$, thereby eliminating $\psi$ and other cubes that fail
for the same reason. Finally, this clause is added to the set of constraints (line 12) and the CEGIS loop continues.
\begin{restatable}[Correctness of {\sc ConstructSeparator}]{theorem}{ConstructSeparator}\label{thm:sep}
{\sc ConstructSeparator} returns the smallest cube 
$\psi\in\mathsf{Cubes}(\mathcal{U}^+)$ such that
$ \Phi^+\Rightarrow\psi$ and  $
\forall j.\ \mathrm{UNSAT}(\psi\land\varphi_j)$, or $\bot$ if no such cube exists.
\end{restatable}



\begin{example}
    Consider the following hypothesis space of three R programs:
{\small
\begin{align*}
\prog_1 &:= \texttt{Mutate(sum := Col(1) + Col(2))} \\
\prog_2 &:= \texttt{Filter(Col(1) < 0) |> Mutate(sum := Col(1) + Col(2))} \\
\prog_3 &:= \texttt{Filter(Col(1) = 0) |> Mutate(sum := Col(1) + Col(2))}
\end{align*}
}

Given an input table, $\prog_1$ computes a new column containing the sum of columns \texttt{1} and \texttt{2}. $\prog_2$ (resp. $\prog_3$) performs the same mutation, but first removes all rows where the value in column \texttt{1} is less than 0 (resp. equal to 0). Since there are only three programs, each program is placed in its own bin $B_i$. Under precondition $\phi := \texttt{cell}_{1,1} = -1 \wedge \texttt{cell}_{2,1} = 0$ (where $\texttt{cell}_{i,j}$ corresponds to the value in the $i$th row and $j$th column), and assuming that the input table contains 2 rows and 2 columns, the strongest postconditions are as follows:

{\footnotesize
\begin{align*}
    \textsf{StrongestPost}(\prog_1, \phi) =& \texttt{ num\_rows} = 2 \wedge \texttt{num\_columns} = 3 \wedge \texttt{output\_cell}_{1,1} = \texttt{cell}_{1,1} \wedge \texttt{output\_cell}_{1,2} = \texttt{cell}_{1,2} \wedge 
    \cdots  \\ 
    \textsf{StrongestPost}(\prog_2, \phi) =& \texttt{ num\_rows} = 1 \wedge \texttt{num\_columns} = 3 \wedge  \texttt{output\_cell}_{1,1} = \texttt{cell}_{2,1} \wedge \texttt{output\_cell}_{1,2} = \texttt{cell}_{2,2} \wedge \cdots\\ 
    \textsf{StrongestPost}(\prog_3, \phi) =& \texttt{ num\_rows} = 1 \wedge \texttt{num\_columns} = 3 \wedge  \texttt{output\_cell}_{1,1} = \texttt{cell}_{1,1} \wedge \texttt{output\_cell}_{1,2} = \texttt{cell}_{1,2} \wedge \cdots 
\end{align*}
}

These postconditions constrain the shapes of the output tables and the values therein. Suppose our atom universe $\Upost$ contains all of the atoms in the postconditions. To construct a separator for $B_1$, we must select interpolants that rule out $B_2$ and $B_3$. In this case, both $B_2$ and $B_3$ may be ruled out by the single atom $\texttt{num\_rows} = 2$, since both $\prog_2$ and $\prog_3$ are guaranteed to filter a row. For $B_2$, we need to find the simplest constraint that rules out both $B_1$ and $B_3$, while capturing the behavior of $B_2$. Selecting $\texttt{ num\_rows} = 1 \wedge \texttt{output\_cell}_{1,1} = \texttt{cell}_{2,1}$ as the separator satisfies both constraints and is the simplest such predicate. Finally, for $B_3$, the simplest such predicate is $\texttt{num\_rows} = 1 \wedge \texttt{output\_cell}_{1,1} = \texttt{cell}_{1,1}$. Thus, for this example, our method would pose the  multiple-choice query shown in Figure~\ref{fig:mc-ex}. 

\end{example}

\subsection{Merging Clusters}\label{sec:merge-clusters}

\begin{algorithm}[t]
\caption{\textsc{MergeClusters}}
\label{alg:merge-clusters}
\begin{algorithmic}[1]
\Require Clusters $C_1,\dots,C_N$, number of answers $k$, precondition $\precond$
\Ensure \Call{EvaluateObjective}{$\mathcal{F}, \precond$} $\geq O_{\mathrm{min}}$, or $\mathcal{F}$ is optimal 
\vspace{0.05in}
\Function{MergeClusters}{$\{C_i\}_{i=1}^N,\,k, \precond$}
\State $\mathcal{F} \gets \Call{LbPartition}{\{C_i\}_{i=1}^N,\,k}$ 
\State $O_{\mathrm{init}} \gets \Call{EvaluateObjective}{\mathcal{F}, \precond}$ \label{line:seed-eval} 
\State $(\mathcal{F},\,O_{\mathrm{best}}) \gets \Call{BranchAndBound}{\{C_i\}_{i=1}^N,\,\varnothing,\,O_{\mathrm{init}}, \precond}$ \label{line:branch-and-bound} 
\State \Return $\mathcal{F}$
\EndFunction
\vspace{0.05in}
  \Function{LbPartition}{($\{C_i\}_{i=1}^N$, $k$)} 
  \For{$i=1$ to $N$}
    \State $\psi_i \gets \Call{ConstructSeparator}{C_i,  \{ \cup_{ j \neq i} C_j \}, \Upost}$
    \State $w_i \gets |C_i| + \lambda\,\complexity(\psi_i)$
  \EndFor
  \State \Return $ \Call{LB-BinPack-Solver}{\{w_i\}_{i=1}^N,\,k}$
  \EndFunction  



\end{algorithmic}
\end{algorithm}

We now discuss the \textsc{MergeClusters} procedure (Algorithm~\ref{alg:merge-clusters}) for producing a partition of the hypothesis space into $k$ clusters. 
This algorithm takes as input the initial \emph{fine-grained partition} $C_1, \ldots, C_N$ (induced by the strongest postconditions $\gamma_i$) and produces \emph{coarse-grained} partition $\{B_1, \ldots, B_k\}$ for a fixed $k$, which we represent by a \emph{partition mapping} $\mathcal{F}:[N]\to [k]$, where $\mathcal{F}(i)=j$ indicates that cluster $C_i$ is included in bin $B_j$. For a fixed precondition and hypothesis space, any partition mapping $\mathcal{F}$ naturally induces a multiple-choice query, defined as follows:
\begin{definition}[Induced query]\label{def:induced-query}
Given hypothesis space $\mathcal{H}$, precondition $\phi$, and partition mapping $\mathcal{F}$, the \emph{query induced by} $(\mathcal{F},\phi^*)$ is
$(\phi,\,\psi_1(\mathcal{F}),\dots,\psi_k(\mathcal{F})) $, where
{
\[
\begin{aligned}
P_j(\mathcal{F}) & =  \{\,h \in \hs \mid \exists\, i \;.\; \mathcal{F}(i)=j \ \wedge\ h \in C_i\,\},
\qquad
\mathcal{N}_j(\mathcal{F}) 
= \left\{\, N_r \ \Bigm\vert\ \ r\neq j,\ N_r= \bigcup\nolimits_{\substack{\mathcal{F}(i)=r}} C_i \right\}
\\[4pt]
& \ \ \ \ \ \ \ \ \ \ \ \ \ \ \ \ \ \ \ \ \ \ \ \ \ \ \ \psi_j(\mathcal{F}) 
:= \textsc{ConstructSeparator}\bigl(P_j(\mathcal{F}),\,\mathcal{N}_j(\mathcal{F}),\,\Upost\bigr).
\end{aligned}
\]}
\end{definition}
\noindent Intuitively, each bin $B_j$ is represented by a separator $\psi_j$ that distinguishes it from all others. Our goal is to generate an informative and interpretable query, which corresponds to generating a partition mapping that results in balanced clusters and interpretable separators, respectively.

The key challenge is that evaluating the quality of a partition requires computing its separators, but this makes the optimization problem computationally intractable: the space of possible partitions is exponential in $N$, and each candidate partition requires invoking \textsc{ConstructSeparator}, which is already solving an NP-hard problem. We address this challenge using a two-phase approach. In the first phase (\textsc{LbPartition}), we optimize a proxy objective that approximates the true objective, resulting in a \emph{proxy partition} that is far more efficient to compute since it does not require explicitly constructing 
separators. Then, the second phase performs a branch-and-bound search starting from the proxy partition. Intuitively, since the proxy partition is a high-quality starting point, the branch-and-bound procedure can efficiently prune the search space.



The first phase is based on the insight that the quality of a partition can be estimated without computing all pairwise separators. Instead, we approximate the true objective by constructing, for each initial fine-grained cluster, a \emph{one-vs-all} separator that distinguishes it from the union of all other clusters. This strategy provides a good signal about how easily each cluster can be separated from the rest, allowing us to \emph{estimate} the cost of potential merges. Specifically, we formulate a proxy objective that takes into account both (1) how balanced a partition is, and (2) its estimated separator complexity cost, obtained by aggregating the one-vs-all separator costs.

In more detail, we formulate this proxy optimization problem as a \emph{load-balanced bin-packing} task, where the goal is to distribute items of varying weights into a fixed number of bins so that no bin becomes disproportionately heavy. 
In our setting, each fine-grained cluster $C_i$ plays the role of an item, and each answer option corresponds to a bin. 
The ``weight'' of each item reflects both the number of programs it contains and the complexity of its separator:
\[
w_i \;=\; |C_i| + \lambda\,\complexity(\psi_i),
\quad\text{where}\quad
\psi_i = \textsc{ConstructSeparator}(C_i,  \{ \cup_{ j \neq i} C_j \}, \Upost)
\]
Here, $|C_i|$ penalizes clusters that contain many candidates and $\mathsf{C}(\psi_i)$ is the syntactic complexity of the one-vs-all separator distinguishing $C_i$ from all other clusters. 
The optimization seeks a mapping of clusters to $k$ bins that minimizes the load of the heaviest bin, thereby achieving two desirable properties: 
(i) no single answer option dominates the hypothesis space, and 
(ii) each option corresponds to clusters that can be separated using simple conditions. We let $\mathcal{F}$ denote the proxy partition mapping constructed by {\sc LbPartition}.

In the second phase, our algorithm performs branch-and-bound search over the space of partition mappings with $\mathcal{F}$ as the starting point. First, the {\sc EvaluateObjective} procedure in line 3 computes the induced query for $\mathcal{F}$, and then  computes its true objective value $O_{\mathrm{init}}$ according to Definition~\ref{def:decomposed-query-selection}.
The branch-and-bound implementation is standard. It organizes the search space as a tree, where each internal node corresponds to a partial assignment of fine-grained clusters to bins, and each leaf represents a complete partition. It uses an admissible heuristic on the best achievable objective; for efficiency, it does not need to call {\sc ConstructSeparator}. We provide details in Appendix~\ref{sec:branch-and-bound}.

\section{Implementation}
\label{sec:impl}

We have implemented the proposed active learning technique as a new tool called \toolname, written in Python. \toolname  uses the Z3 SMT solver~\cite{z3} for checking satisfiability and solving optimization problems.

\rev{
\paragraph{\textbf{Translating queries to natural language.}}
Our implementation leverages \texttt{gpt-4o} to translate logical queries into natural language (NL) descriptions. We utilize few-shot prompting \cite{brown2020language}, providing the LLM with a query along with a small set of in-context examples. Because this translation is not formally verified, a mistranslation could in principle cause the user to select an incorrect answer. In practice, however, we find this risk to be negligible because the queries are expressed in first-order logic over a small, well-typed predicate set -- this is a setting where LLMs are highly reliable ~\cite{photoscout}. 
}

\paragraph{\textbf{Instantiating in new domains.}}  
The design of \toolname is domain-agnostic and can be instantiated for different synthesis settings by providing three components:
(1) a synthesizer for generating the initial hypothesis space,
(2) an analysis engine for computing pre- and postconditions, and
(3) the initial universes of pre- and postcondition predicates.
Depending on the domain, the analysis engine may either apply standard invariant-generation techniques to reason about iterative or higher-order constructs (e.g., \texttt{fold}), or unroll these constructs to a fixed bound and compute pre- and postconditions on the resulting loop- and recursion-free programs. In the latter case, \toolname guarantees semantic equivalence only up to the chosen unrolling depth.
For the domains used in our evaluation, we adopt this bounded-unrolling strategy and manually verify the correctness of the final synthesized program.
\rev{The construction of the predicate universes is straightforward. 
Precondition predicates in $\mathcal{U}^-$ follow the shape of predicates already exposed by the DSL, while postcondition predicates in $\mathcal{U}^+$ encode possible effects of the program on the input.
For example, in the image editing domain, $\mathcal{U}^-$ includes predicates such as \texttt{HasLabel(obj, Person)} and \texttt{HasRelation(obj\_1, obj\_2, NextTo)}, whereas $\mathcal{U}^+$ includes predicates such as \texttt{Blurred(obj)} and \texttt{Cropped(obj)} that describe observable output behavior.
More generally, $\mathcal{U}^-$ captures properties of inputs that may appear in synthesized preconditions, while $\mathcal{U}^+$ captures properties of outputs or input-output relationships that may appear in answer choices.
}





\paragraph{{\textbf{Approximating the objective.}}}  
While the algorithms described in Section~\ref{sec:algo} compute the optimal query as defined in
Definition~\ref{def:decomposed-query-selection}, our implementation employs two practical
approximations to ensure tractability. 
First, following prior work~\cite{samplesy, smartlabel}, we uniformly sample a  subset of 
programs from the hypothesis space each round and compute queries that are optimal with respect to this subset,
rather than over the full hypothesis space. 
Second, our implementation of {\sc MergeClusters}  invokes  branch-and-bound search {only} when the solution produced by 
\textsc{LbPartition} exhibits  complexity exceeding a threshold. 
These approximations preserve the intent of optimal query selection while keeping the computation
efficient in practice. 




\section{Evaluation}
\label{sec:eval}

In this section, we describe the results of our experimental evaluation, which aims to answer the following research questions: 

\begin{itemize}[leftmargin=*]
\item \textbf{RQ1: Accuracy.} How does \toolname compare against state-of-the-art active learning baselines in terms of accuracy?
\item \textbf{RQ2: Interpretability.} How do users perform when answering the multiple-choice queries posed by \toolname{} compared to traditional input–output labeling questions?
\item \textbf{RQ3: Efficiency.} How does \toolname{} compare against baselines in terms of efficiency (e.g., number of interaction rounds, query generation time)?
\item \textbf{RQ4: Ablations.} How important are our key algorithmic ingredients in reducing active learning runtime and generating simple queries? 
\end{itemize}

\subsection{Application Domains}
To address our research questions, we evaluate \toolname and all baselines on 157 tasks drawn from four domains studied in prior work: data wrangling~\cite{morpheus}, JSON transformations~\cite{faery}, batch image editing~\cite{imageeye}, and image search~\cite{photoscout}. All of these domains involve rich  input types, namely tables, trees, and images, and therefore provide a meaningful basis of evaluation for our approach. We describe each domain in more detail below.

\paragraph{Table transformations} Our first application domain, \tables, consists of 80 challenging \emph{data wrangling} tasks considered in prior work~\cite{morpheus, feng2018program}. Each task involves transforming one or more input tables into a target table using a DSL that is inspired by R's \texttt{dplyr} and \texttt{tidyr} libraries. Typical transformations include reshaping data between ``wide'' and ``long'' formats, consolidating multiple tables, and performing grouped aggregations or joins.

\paragraph{JSON transformations.}
Our second application domain, \trees, consists of 15  JSON transformation tasks  from prior work~\cite{faery}. Each benchmark involves converting hierarchical input trees into structurally distinct outputs using a domain-specific language that supports node creation, filtering, and restructuring through declarative tree combinators. These transformations capture a  range of structural manipulations, such as field extraction, flattening, and conditional reorganization.

\paragraph{Batch image editing} 
Our third application domain, \imageedit, involves image editing tasks considered in prior work~\cite{imageeye,smartlabel}. Unlike the two previous domains, these tasks are neurosymbolic and require synthesizing programs in a DSL that involves neural networks for image classification and segmentation.  Each task specifies a high-level edit (e.g., “brighten faces of bride and groom” “crop all people playing guitar'') that must be realized by composing various neural components with symbolic operators. These benchmarks come with collections of real-world images.

\paragraph{Image search} Our fourth and final application domain, \imagesearch,  involves image search tasks from prior work~\cite{smartlabel,photoscout} where the goal is to filter a subset of images that have a certain property, such as ``contain a dog and a cat next to each other,'' or ``contain a person riding a bike while wearing a helmet.'' Similar to the previous domain, these tasks  require synthesis in a neurosymbolic DSL that has neural constructs for image segmentation and object classification.

\subsection{Baselines}

To meaningfully evaluate our proposed approach, we compare \toolname{} against state-of-the-art active learning techniques from recent work. Because no single prior method applies uniformly across both symbolic and neurosymbolic domains, our evaluation considers two groups of baselines.

\paragraph{Purely symbolic domains.} For purely symbolic domains (\tables and \trees), we compare against {\bf \textsc{SampleSy}}~\cite{samplesy} and {\bf \textsc{LearnSy}}~\cite{learnsy}. \textsc{SampleSy} selects the input whose \emph{worst-case} label would eliminate the largest fraction of the hypothesis space, while \textsc{LearnSy} generalizes this strategy by introducing a probabilistic model of program equivalence that estimates the \emph{expected} information gain from labeling a particular input. 

\paragraph{Neurosymbolic domains.}  
For neurosymbolic domains (\imageedit and \imagesearch), we compare \toolname{} against {\bf \textsc{SmartLabel}}~\cite{smartlabel}, which is the only prior technique that provides formal guarantees against eliminating the intended program in the presence of neural uncertainty. \textsc{SmartLabel} extends \textsc{SampleSy}'s objective to neurosymbolic settings by incorporating conformal prediction~\cite{confpred}, producing set-valued outputs that bound the true label with high confidence.

\subsection{Experimental Setup and Methodology}

We evaluate \toolname, along with all baselines and ablations, using the following methodology. We begin by generating two initial input–output examples and constructing a hypothesis space~$\hs$ of programs consistent with the examples, using either an enumerative synthesizer or an LLM-based agent. The hypothesis space always includes the ground-truth program $P^*$; however, it also includes many other programs that conform to the initial examples, but semantically differ from the ground truth. Once such a hypothesis space $\hs$ is constructed, we  perform active learning over~$\hs$, using an oracle that provides correct responses to each query. To account for potential variability from the initial random I/O examples, we repeat all experiments five times, each with a different seed, and report the mean and standard deviation of the outcomes.


\begin{wraptable}{r}{0.52\linewidth}  
\vspace{-0.1in}
\centering
\footnotesize
\caption{Details about (1) the number of tasks,  (2)  average input space size, (3) initial hypothesis space size, (4) average program size, and (5) the percentage of program pairs in the initial hypothesis space that are semantically distinguishable.}
\label{tab:stats}
\vspace{-0.1cm}
\setlength{\tabcolsep}{3pt} 
\begin{tabular}{lccccc}
\toprule
\textbf{Domain} & \# & \textbf{|Inputs|} & \textbf{|$\hs$|} & \textbf{AST size} & \textbf{\rev{\% Pairs Dist.}}  \\
\midrule
\tables      & 80 & 75.0  & 50.0   & 14.6 & \rev{58.1\%} \\
\trees       & 15 & 375.0 & 95.6   & 11.7 & \rev{86.4\%} \\
\imageedit   & 37 & 271.7 & 1096.4 & 15.7 & \rev{87.3\%} \\
\imagesearch & 25 & 257.4 & 318.8  & 16.0 & \rev{92.9\%} \\
\midrule
\textbf{Overall} & 157 & 244.8 & 390.2 & 14.5 & \rev{73.2\%} \\
\bottomrule
\end{tabular}
\vspace{-0.2cm}
\end{wraptable}

Table~\ref{tab:stats} summarizes key statistics for our benchmark domains, including the number of tasks, average program size, and average  cardinality of the initial hypothesis space. While \toolname{} does not require access to a predefined set of inputs on which the target program is evaluated, all of the baseline methods do. In particular, these baselines are parameterized by a finite \emph{input space} from which candidate queries are drawn. Thus, Table~\ref{tab:stats} also shows the size of the input space used for evaluating the baselines.  Specifically, for the \imageedit and \imagesearch domains, we adopt the same input spaces used in prior work~\cite{smartlabel}. In contrast, for the \tables{} and \trees{} domains, no standard input dataset is available; we construct synthetic input spaces by sampling schema-compatible inputs that satisfy the grammar of each domain.\footnote{To make the comparison fair and meaningful, we must choose the size of the input space carefully: including too few inputs hurts accuracy, while including too many increases the time needed to compute each query. Following findings from the HCI literature on user tolerance for interactive wait times~\cite{user-tolerance}, we select the largest input space for which all methods complete query generation within 10 seconds, which has been reported as the upper bound of acceptable response latency for sustained user engagement. 
} 

\rev{To further characterize the complexity of query generation, Table~\ref{tab:stats} also reports the fraction of program pairs in the initial hypothesis space that are semantically distinguishable. This metric captures how easily the hypothesis space can be pruned by informative queries: when many pairs are distinguishable, the system can more readily construct queries that expose behavioral differences among candidates, whereas greater semantic overlap means that many candidates behave identically on large parts of the input space and are therefore harder to separate. Overall, 73.2\% of program pairs are distinguishable, although this fraction varies across domains, with tasks in the  \tables{} domain having the lowest distinguishability rate (58.1\%).}


All experiments are executed on a 2022 MacBook Pro with an 8-core M2 processor and 8 GB of RAM, using a timeout limit of 300 seconds per task.

\subsection{Evaluation of Accuracy}\label{sec:accuracy}

Our first research question investigates whether \toolname{} improves the accuracy of interactive program synthesis compared to existing active learning techniques. 
For each method, we measure the percentage of benchmarks for which the synthesized program is \emph{semantically equivalent} to the ground-truth program~$P^*$. 
A benchmark counts as solved only if the synthesized program is manually verified to behave identically to~$P^*$ on all possible inputs. 
Tables~\ref{tab:accuracy1} and~\ref{tab:accuracy2} report these results for the symbolic (\tables, \trees) and neurosymbolic (\imageedit, \imagesearch) domains.
Across all domains, \toolname{} solves 100\% of benchmarks. For the \tables and \trees domains, \textsc{SampleSy} and \textsc{LearnSy} solve 77.5\% and 76.2\% benchmarks, respectively, and for the neurosymbolic domains, \textsc{SmartLabel} solves 85.5\% of the benchmarks.

\begin{figure}[t]
\begin{minipage}{0.49\linewidth}
\captionof{table}{Experimental results comparing accuracy of active learning techniques in symbolic domains.}
\label{tab:accuracy1}
\footnotesize
\centering
\begin{tabular}{lccc}
\toprule
\textbf{Domain} & \textbf{\textsc{Socrates}} & \textbf{\textsc{SampleSy}} & \textbf{\textsc{LearnSy}} \\
\midrule
\tables & {\bf 100.0\% $\pm$ 0\%} & 80.3\% $\pm$ 4.3\% & 78.5\% $\pm$ 5.4\% \\
\trees  & {\bf 100.0\% $\pm$ 0\%} & 62.7\% $\pm$ 7.6\%  & 64.0\% $\pm$ 6.0\%\\ 
\midrule
{Overall} & {\bf 100.0\% $\pm$ 0\%} & 77.5\% $\pm$ 2.5\% & 76.2\% $\pm$ 3.8\% \\
\bottomrule
\end{tabular}
\end{minipage}
\hfill
\begin{minipage}{0.49\linewidth}
\captionof{table}{Experimental results comparing accuracy of active learning techniques in neurosymbolic domains.}
\label{tab:accuracy2}
\footnotesize
\centering
\begin{tabular}{lcc}
\toprule
\textbf{Domain} & \textbf{\textsc{Socrates}} & \textbf{\textsc{SmartLabel}} \\
\midrule
\imageedit   & {\bf 100\% $\pm$ 0\%} & 84.9\% $\pm$ 2.4\% \\
\imagesearch & {\bf 100\% $\pm$ 0\%} & 86.4\% $\pm$ 8.3\% \\
\midrule
{Overall} & {\bf 100\% $\pm$ 0\%} & 85.5\% $\pm$ 2.3\% \\
\bottomrule
\end{tabular}
\end{minipage}
\vspace{-.3in}
\end{figure}

\paragraph{{Failure analysis for baselines.}}
Most baseline failures stem from their reliance on a fixed input set for checking equivalence. These methods iteratively refine the hypothesis space until all remaining candidates are observationally indistinguishable on that set and then return one randomly sampled program. However, observational equivalence on a finite input set does not imply true semantic equivalence, and the selected program may therefore diverge from the ground truth. As discussed in Section~\ref{sec:impl}, the equivalence guarantee provided by \toolname is also not universal, due to the bounded unrolling used to compute pre- and postconditions. Despite this practical limitation, all programs synthesized by \toolname were manually verified to be semantically equivalent to the ground truth. In contrast, the baselines frequently produce programs that match the ground truth on all observed examples but behave differently on unseen inputs -- an inherent limitation of any active-learning approach restricted to a fixed input set. Finally, \textsc{SmartLabel} sometimes fails for a distinct reason: although its conformal predictor provides statistical coverage guarantees, there remains a small probability that the ground-truth label falls outside the predicted confidence set. In the neurosymbolic domains, a small fraction of failures can be attributed to this cause.

\vspace{0.1in}
\idiotbox{RQ1}{\toolname{}  converges to the ground-truth program in all experimental benchmarks, while the baselines fail to find the intended program for roughly 25\% of the programs in the symbolic domains and 15\% in the neurosymbolic domains. }

\subsection{Evaluation of Interpretability}\label{sec:user-study}

 To answer our second research question, we conducted a user study to evaluate interpretability. This study measures how easily users can understand and answer multiple choice (MC) queries compared to the input–output (I/O) queries produced by existing active learning techniques.


\paragraph{\textbf{Setup.}} We recruited 18 participants, each of whom completed two active learning tasks per application domain: one using \toolname and one using a representative active learning baseline (\textsc{SampleSy} for \tables and \trees, and \textsc{SmartLabel} for \imageedit and \imagesearch). Participants were undergraduate and graduate students in computer science, all with prior programming experience but no familiarity with the evaluated systems. 

To keep the total session length under 45 minutes and reduce cognitive fatigue, each participant was assigned tasks from three of the four domains. Each task was randomly drawn from a pool of five benchmarks, and both the order of domains and the assignment of tools were randomized to mitigate ordering effects.
For each task, participants were first shown a short textual description and an I/O example, followed by a two-minute familiarization period during which they could review the materials and ask clarifying questions. They then answered all queries issued by the corresponding active learning tool. We recorded both the time  to answer each query and whether the answer matched the ground truth.

\begin{wraptable}{r}{0.45\textwidth}
\vspace{-0.2in}
\centering
\caption{User study comparing the interpretability of multiple-choice (MC)  and  I/O queries.}
\label{tab:userstudy}
\footnotesize
\setlength{\tabcolsep}{5pt}
\begin{tabular}{@{}lcccc@{}}
\toprule
\textbf{Domain} &
\multicolumn{2}{c}{\textbf{Time (s)}} &
\multicolumn{2}{c}{\textbf{Accuracy (\%)}} \\
\cmidrule(lr){2-3} \cmidrule(lr){4-5}
& \textsc{MC} & \textsc{I/O} & \textsc{MC} & \textsc{I/O} \\
\midrule
\tables      & {\bf 60.5} & 98.8 & {\bf 93.8} & 86.7 \\
\trees       & {\bf 59.9} & 70.5 & {\bf 77.8} & 54.1 \\
\imageedit   & 45.5 & {\bf 38.6} & {\bf 88.2} & 57.9 \\
\imagesearch & 25.6 & {\bf 21.2} & {\bf 92.3} & 57.8 \\
\midrule
\textsc{Overall} & {\bf 48.2} & 55.9 & {\bf 88.5} & 64.1 \\
\bottomrule
\end{tabular}
\vspace{-0.2cm}
\end{wraptable}
\paragraph{\textbf{Results}} Table~\ref{tab:userstudy} summarizes the results of our user study. Across all domains, users answered MC questions slightly faster and 38\% more accurately than I/O questions. The overall accuracy improvement is statistically significant ($p = 4.35\times10^{-5}$, paired-sample t-test). Using the same test, response time differences were statistically significant only in the \tables domain ($p = 0.043$), where users answered MC questions 38 seconds faster on average. This difference likely reflects the additional overhead of manually typing tables for I/O questions, which was considerably more time-consuming than selecting a multiple-choice answer. In the \trees domain, users answered MC questions 44\% more accurately and 10 seconds faster than I/O questions, suggesting that the structured presentation of MC queries helped users reason about hierarchical data more effectively. 
In the \imageedit and \imagesearch domains, MC questions yielded 52\% and 60\% higher accuracy, respectively, despite slightly longer response times. For these domains, I/O questions required users to inspect images and identify objects, leading to frequent misidentifications or omissions. In contrast, MC questions provided textual descriptions and a small set of candidate options, which seems to have reduced ambiguity and human error. Appendix~\ref{sec:errors} illustrates representative  errors users make when answering queries. 

\vspace{0.1in}
\idiotbox{RQ2}{Users answered  multiple-choice queries significantly more accurately than the input–output queries, while maintaining comparable response times.}

\subsection{Evaluation of Efficiency}\label{sec:efficiency}

Having established that \toolname{}  improves both synthesis and query accuracy, we next examine whether this improvement comes at the cost of efficiency. We measure efficiency along two dimensions: (1) the number of interaction rounds required to converge, and (2) the average time needed to generate each query. Table~\ref{tab:runtime} summarizes these results.

Across all domains, \toolname{} is competitive in terms of efficiency. In the \tables{} domain, it converges in roughly the same number of rounds as \textsc{LearnSy} and fewer than \textsc{SampleSy}, while maintaining lower average query-generation time. In the neurosymbolic domains (\imageedit{} and \imagesearch), \toolname{} matches or outperforms \textsc{SmartLabel} in runtime and requires roughly the same number of interaction rounds despite its stronger semantic guarantees. The only setting where \toolname{} is less efficient is the \trees{} domain.  However, this is explained by the fact that the baselines have the lowest accuracy in this domain: while they take approximately 3 fewer rounds of user interaction, their accuracy is less than 65\% (compared to 100\% of \toolname).

Our runtime measurements indicate that \toolname{}’s computational cost is dominated by LLM inference (using \texttt{gpt-4o}), which is used to translate logical queries into natural language. To isolate this effect, Table~\ref{tab:runtime} also reports results for \toolname{\sc -InstantLLM}, which omits LLM processing time. This configuration is consistently faster than all baselines across every domain. Although a few \tables{} benchmarks exhibit outliers that inflate the mean, the median query-generation time for \toolname{\sc -InstantLLM} remains under two seconds. It is worth noting that the cost of LLM inference primarily reflects model throughput and network variability, and it can be substantially reduced through improved deployment infrastructure ~\cite{zhou2024survey} (e.g., provisioned throughput or specialized inference accelerators) without any modification to \toolname{}’s core technique.

\paragraph{Impact of different components on runtime}
\rev{
Another interesting empirical question is where \toolname{} spends the majority of its query-generation time in practice. As noted above, LLM inference for translating queries into natural language accounts for a significant portion of runtime (roughly 42\%), but this cost is not a fundamental algorithmic bottleneck. 
Figure ~\ref{fig:runtime} shows how the remaining  components contribute to query-generation time.  }

\begin{wrapfigure}{r}{0.35\textwidth}
  \begin{center}
  \vspace{-.2in}
        \includegraphics[
        width=0.34\textwidth,
        trim=0 0 0 2cm,
        clip
    ]{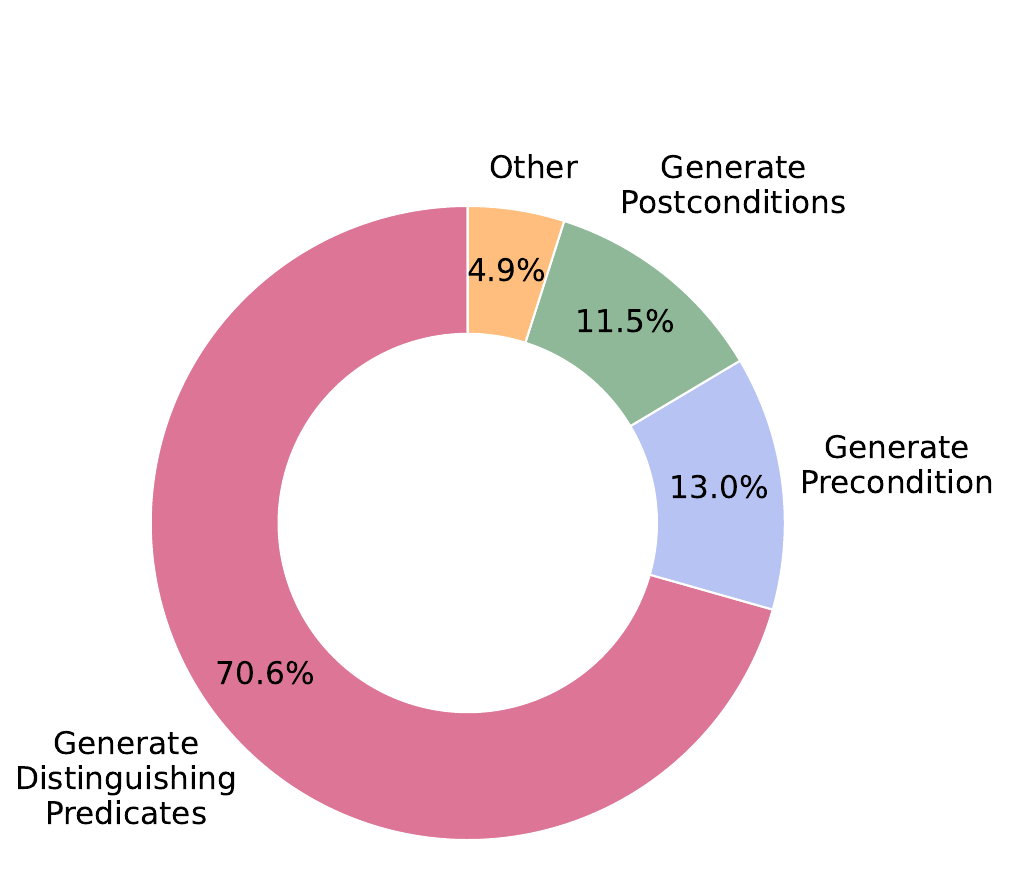}
  \end{center}
  \vspace{-0.13in}
    \caption{Breakdown of active learning runtime, excluding LLM inference time.}\label{fig:runtime}
  \vspace{-0.2in}
\end{wrapfigure}

\rev{
Across all domains, computing distinguishing predicates accounts for 70.6\% of runtime, precondition synthesis accounts for 11.5\%, and postcondition generation accounts for 13.0\%. This breakdown varies by domain. In the \tables{} domain, computing distinguishing predicates dominates runtime (81\%). In the \imageedit{} and \imagesearch{} domains, precondition synthesis becomes the primary cost, accounting for 60–66\% of runtime. In the \trees{} domain, query generation time is distributed more evenly across the three components.
}

\rev{
This trend reflects the greater symbolic reasoning burden in the \tables{} domain, where programs compose multiple table transformations and exhibit the lowest semantic distinguishability (58.1\%, Table~\ref{tab:stats}). Furthermore,  weakest precondition size is largest in this domain.
}

\idiotbox{RQ3}{\toolname{} achieves stronger semantic guarantees with only a marginal increase in rounds of user interaction and comparable or faster query computation time.}

\begin{table}[t]
\centering
\caption{Experimental results comparing \toolname against  active learning baselines. For each technique, we report the averages of (1) number of interaction rounds, (2) time per round, and (3) total tool runtime.}
\small
\begin{tabular}{c l c c c}
\toprule
\textbf{Domain} & \textbf{Technique} & \textbf{\# Rounds} & \textbf{Time per Round (s)} & \textbf{Total Tool Time (s)}\\
\midrule
\multirow{4}{*}{\tables}
  & \toolname & 4.4 $\pm$ 0.1 & 7.8 $\pm$ 0.4 & 34.5 $\pm$ 2.3 \\
  & \toolname-\textsc{InstantLLM} & 4.4 $\pm$ 0.1 & 5.7 $\pm$ 0.4 & 25.2 $\pm$ 1.8 \\
  & \textsc{SampleSy} & 5.1 $\pm$ 0.3 & 9.6 $\pm$ 0.3 & 48.4 $\pm$ 2.0  \\
  & \textsc{LearnSy}    & 4.2 $\pm$ 0.1 & 8.5 $\pm$ 0.2 & 35.8 $\pm$ 1.9  \\
\midrule
\multirow{4}{*}{\trees}
  & \toolname & 5.6 $\pm$ 0.2 & 5.9 $\pm$ 0.5 & 33.4 $\pm$ 2.5 \\
  & \toolname-\textsc{InstantLLM} & 5.6 $\pm$ 0.2 & 0.9 $\pm$ 0.0 & 5.1 $\pm$ 0.3 \\
  & \textsc{SampleSy} & 2.6 $\pm$ 0.1 & 9.5 $\pm$ 0.8 & 25.2 $\pm$ 2.6  \\
  & \textsc{LearnSy}    & 3.4 $\pm$ 0.3 & 4.8 $\pm$ 0.2 & 16.3 $\pm$ 1.3  \\
\midrule
\multirow{2}{*}{\imageedit}
  & \toolname & 4.1 $\pm$ 0.6 & 3.1 $\pm$ 0.2 & 12.9 $\pm$ 2.3 \\
  & \toolname-\textsc{InstantLLM} & 4.1 $\pm$ 0.6  & 0.9 $\pm$ 0.2 & 3.8 $\pm$ 0.9 \\
  & \textsc{SmartLabel} & 3.7 $\pm$ 0.5 & 5.5 $\pm$ 0.6 & 20.6 $\pm$ 6.5 \\
\midrule
\multirow{2}{*}{\imagesearch}
  & \toolname           & 4.0 $\pm$ 0.6 & 2.9 $\pm$ 0.2 & 11.7 $\pm$ 2.9 \\
  & \toolname-\textsc{InstantLLM} & 4.0 $\pm$ 0.6 & 0.7 $\pm$ 0.1 & 2.9 $\pm$ 0.7 \\
  & \textsc{SmartLabel} & 4.5 $\pm$ 0.8 & 6.3 $\pm$ 2.2  & 27.4 $\pm$ 7.3 \\
\bottomrule
\end{tabular}
\label{tab:runtime}
\vspace{-0.1in}
\end{table}

\subsection{Ablation Studies}\label{sec:ablation}

To quantify the contribution of each major algorithmic component in \toolname, we conduct ablation studies that remove or modify one component at a time while keeping others fixed. Specifically, we evaluate the following three variants:
\begin{itemize}[leftmargin=*]
    \item \textbf{\textsc{Socrates-SimplePre}}, which selects any distinguishing constraint between two non-equivalent programs as the precondition. This ablation isolates the impact of \textsc{GetBestPrecondition}.
    \item \textbf{\textsc{Socrates-SimpleSep}}, which computes a separator by taking the  disjunction of the strongest postconditions within each cluster and then simplifying the resulting formula~\cite{simplify}. 
    This ablation isolates the impact of \textsc{ConstructSeparator}.
    \item \textbf{\textsc{Socrates-RandCluster}}, which constructs roughly equal-sized clusters by  assigning each program to a randomly chosen cluster. After forming the clusters, it constructs a multiple-choice question by calling {\sc ConstructSeparator},  merging clusters as needed if a separator does not exist. This ablation isolates the impact of our {\sc MergeClusters} procedure. 
\end{itemize}

\noindent

\begin{table}[t]
\caption{Summary of ablation study results.}
\label{tab:ablation-summary}
\vspace{-0.1in}
\centering
\renewcommand{\arraystretch}{1.3} 
\begin{tabular}{@{}>{\centering\arraybackslash}m{0.28\linewidth} m{0.64\linewidth}@{}}
\toprule
\textbf{Variant} & \textbf{Key impact} \\
\midrule

\textsc{Socrates-SimplePre} &
Increases the number of rounds by about $2\times$, while modestly improving query generation time and query complexity. \\
\hline

\textsc{Socrates-SimpleSep} &
Nearly triples postcondition complexity, with minimal effect on other metrics. \\
\hline

\textsc{Socrates-RandCluster} &
Almost doubles the number of rounds and increases postcondition complexity by more than $2.5\times$. \\
\bottomrule
\end{tabular}
\renewcommand{\arraystretch}{1.0} 
\vspace{-0.1in}
\end{table}

\begin{table}[t]
\centering
\caption{\rev{Detailed ablation study results. The first row reports absolute values for \toolname{}. Remaining rows report percentage change relative to \toolname{}.}}
\label{tab:ablation}
\small
\vspace{-0.1in}
\setlength{\tabcolsep}{6pt}
\begin{tabular}{@{}lcccc@{}}
\toprule
\textbf{Variant} 
  & \textbf{\# Rounds} 
  & \textbf{Time/Round (s)} 
  & \textbf{Precond.} 
  & \textbf{Postcond.} \\
\midrule
\toolname 
  & 4.4 $\pm$ 0.1 
  & 3.6 $\pm$ 0.3 
  & 21.4 $\pm$ 0.1 
  & 6.9 $\pm$ 0.0 \\ 
\midrule
\toolname{\sc -SimplePre} 
  & \rev{+102\%}\upar 
  & \rev{-61\%}\downar 
  & \rev{-33\%}\downar
  & \rev{-39\%}\downar \\
\toolname{\sc -SimpleSep} 
  & \rev{-2\%}\noar 
  & \rev{+11\%}\noar 
  & \rev{+0\%}\noar
  & \rev{+191\%}\upar \\
\toolname{\sc -RandCluster} 
  & \rev{+84\%}\upar
  & \rev{+17\%}\noar
  & \rev{+4\%}\noar
  & \rev{+154\%}\upar \\
\bottomrule
\end{tabular}
\vspace{-0.2in}
\end{table}

Tables \ref{tab:ablation-summary} and~\ref{tab:ablation} report the results of this ablation study. Specifically, Table \ref{tab:ablation-summary} summarizes the key impact of disabling a given component, and Table~\ref{tab:ablation} provides more detailed statistics about the averages of (1) the number of interaction rounds, (2) query generation time (excluding LLM processing time), (3) precondition complexity, and (4) postcondition complexity (measured by AST size). As we can see, disabling any component of \toolname{} degrades performance along at least one dimension relative to the full configuration. \textsc{Socrates-SimplePre} generates simpler queries quickly, but doubles the number of interaction rounds. This suggests that our OMT-based precondition generation method plays an important role in reducing user effort. \textsc{Socrates-SimpleSep} leaves tool runtime essentially unchanged, but substantially increases postcondition complexity. This indicates that separator construction is important for keeping answer choices easy to understand. Finally, \textsc{Socrates-RandCluster} increases both the number of rounds and the complexity of the postconditions, showing that our clustering strategy helps form bins that are balanced and admit clean separators.

\vspace{0.05in}
\idiotbox{RQ4}{Our key algorithmic ingredients collectively improve the trade-off between query complexity, interaction rounds, and query-generation time.}

\section{Related Work}

\paragraph{\textbf{Active learning for program synthesis.}} Active learning refers to a class of techniques that strategically select data in order to maximize information gain. Originally developed in the context of machine learning~\cite{bshouty1994oracles, bshouty1995exact, schohn2000less, dasgupta2004analysis, settles2009active, ren2021survey}, active learning has recently been applied to \emph{interactive program synthesis}, where it is used to disambiguate candidate programs. In this setting,  the active learner typically queries the user to label an input from a pre-defined input space. 
For example, \emph{FlashProg}~\cite{flashprog} asks the user to label inputs on which candidate programs produce different outputs. Subsequent approaches~\cite{samplesy, learnsy, huang2022neural, smartlabel, abstractexamples} generalize this idea by formulating query selection as an optimization problem wherein the goal is to minimize the number of rounds of user interaction. For example, \emph{SampleSy} and \emph{SmartLabel} both employ greedy minimax algorithms, selecting the question whose worst answer will prune the greatest portion of programs.
These methods iterate until all remaining programs agree on all queries in the input space. However, as demonstrated in Section~\ref{sec:accuracy}, these methods may fail to identify the ground-truth program due to their restricted notion of equivalence. 

An alternative line of work explores symbolic program disambiguation, where the active learner reasons over logical constraints rather than concrete inputs. In this setting, the system constructs a formula characterizing inputs on which  candidate programs produce different outputs and invokes an SMT solver to find an input satisfying that formula. Several prior approaches~\cite{susmit10, scythe, forest} follow this pattern, using the resulting input as a new query for the user. Ramos et al.~\cite{ramos2020program} extend this idea by formulating program disambiguation as a MaxSAT optimization problem that selects an input distinguishing the largest number of program pairs. However, their formulation assumes that program semantics can be fully captured in propositional logic. More broadly, while these kinds of input generation techniques work well for programs over simple data types like integers or strings,  extending them to richer data (e.g., JSON documents or images) remains challenging: constructing a concrete input that satisfies a symbolic constraint is an open problem, especially when constraints involve complex structural or perceptual features. For example, generating an image that satisfies a logical predicate would require constraint-conditioned image synthesis, for which efficient and reliable methods do not yet exist ~\cite{kota2024}.

\paragraph{\textbf{User interaction models for program synthesis.}} Since providing input-output examples in the traditional programming-by-example (PBE) setting is often challenging \cite{pbemistakes, pbemistakes2}, prior work has proposed alternative interaction models that make it easier for users to specify their intent. Several approaches \cite{augexamples, flashprofile} propose interfaces that visualize clusters of data in the input space, as an aid for selecting new input-output examples. Singh and Solar-Lezama \cite{storyboard} propose a method that allows users to specify their task using control-flow diagrams. Other interfaces \cite{peleg2020programming, peleg2} present a candidate program, and let users annotate sub-expressions that should or should not be present in the solution. Drachsler-Cohen et al.~\cite{abstract} present an interactive synthesis system that communicates with the user through abstract examples -- i.e., symbolic input–output pairs describing potentially unbounded sets of examples. The user can either accept an abstract example or provide a concrete counterexample that invalidates it. While this design offers formal guarantees once the accepted abstractions cover the input space, it expects users to reason about symbolic patterns and to submit concrete instances that contradict them.

\paragraph{\textbf{Natural language interaction in synthesis.}} In recent years, there has been much interest in leveraging large language models (LLMs) in synthesis systems. Many prior works allow users to specify their intent with natural language (NL) queries that are processed by an LLM \cite{llm1, llm2, llm3, vipergpt, codegen, llm4}. Since NL is inherently imprecise, some approaches combine NL specifications with other modalities, such as input-output examples \cite{photoscout, chen2020multi, jain2022jigsaw, rahmani2021multi, gavran2020interactive} or demonstrations \cite{li2019pumice, li2017sugilite}.   While writing NL queries is often convenient, these approaches offer no guarantee that the synthesized code is correct across all inputs. In our work, we generate simple natural language queries that are designed to clarify ambiguities in the user's intent. 

\paragraph{\textbf{Predicate abstraction in program synthesis.}} 
Predicate abstraction is a classical technique for approximating program behavior within a finite domain of logical predicates~\cite{predabstr, predabstrC, flanagan2002predicate}. In verification, frameworks such as CEGAR~\cite{cegar} iteratively refine these abstractions to rule out spurious counterexamples, and similar ideas have been applied to program synthesis~\cite{wang2017program, guo2019program, wang2018learning}. In our setting, predicate abstraction serves a different role: it is used to synthesize structured, human-interpretable queries that can be formulated as multiple-choice questions.

\section{Conclusion}
This paper introduced a new paradigm for interactive program disambiguation based on multiple-choice queries, where users choose from a list of high-level behaviors instead of labeling concrete inputs. By formulating each question as a structured logical query over pre- and postconditions, our approach enables more direct communication of semantic intent and enables stronger correctness guarantees. We presented a principled algorithm that decomposes query selection into precondition synthesis and answer generation, realized through a combination of SMT-based optimization, clustering, and separator construction. Our implementation, \textsc{Socrates}, demonstrates that this paradigm leads to substantially higher synthesis accuracy and user response accuracy across both symbolic and neurosymbolic domains, while maintaining competitive efficiency. More broadly, this work highlights that replacing low-level annotation with structured semantic queries can make program disambiguation more practical both for users and synthesis systems. Looking ahead, as large language models are increasingly used for code generation, integrating structured disambiguation mechanisms like ours offers a promising direction for improving the reliability and interpretability of LLM-driven synthesis.
\section{Data-Availability Statement}

The artifact for this paper is available on Zenodo ~\cite{socrates-artifact}.

\bibliography{main}

\newpage
\appendix
\section{Branch-and-Bound Search Algorithm for Merging Clusters} \label{sec:branch-and-bound}

\begin{algorithm}[t]
\caption{\textsc{BranchAndBound}}
\label{alg:branch-and-bound}
\begin{algorithmic}[1]
  \Require Clusters $C_1,\dots,C_N$, incumbent partition $\mathcal{F}_p$, incumbent objective $O_\mathrm{best}$, precondition $\precond$
   \Ensure $\mathcal{F}^*$ is optimal 
\vspace{0.05in}
    \If{$|\mathcal{F}_p| = N$} \Comment{all clusters assigned}
      \State $O \gets \Call{EvaluateObjective}{\mathcal{F}_p, \precond}$ 
      \If{$O > O_{\mathrm{best}}$}  $\mathcal{F}^* \gets \mathcal{F}_p$; \quad $O_{\mathrm{best}} \gets O$
      \EndIf
      \State \Return $(\mathcal{F}^*,\,O_{\mathrm{best}})$
    \EndIf

    \State $U \gets \Call{ComputeBranchUB}{\mathcal{F}_p}$ \Comment{admissible upper bound} 
    \If{$U \le O_{\mathrm{best}}$}
       \Return $(\mathcal{F}^*,\,O_{\mathrm{best}})$ \Comment{prune this subtree}
    \EndIf

    \For{each admissible extension $\mathcal{F}_{p+1}$ of $\mathcal{F}_p$ using $\{C_i\}_{i=1}^N$}
      \State $(\mathcal{F}^*,\,O_{\mathrm{best}}) \gets \Call{BranchAndBound}{\{C_i\}_{i=1}^N,\,\mathcal{F}_{p+1},\,O_{\mathrm{best}}, \precond}$ 
    \EndFor

    \State \Return $(\mathcal{F}^*,\,O_{\mathrm{best}})$
\end{algorithmic}
\end{algorithm}

Algorithm~\ref{alg:branch-and-bound} describes our branch-and-bound-style search procedure. This algorithm induces a search tree where leaf nodes correspond to complete partitions, while internal nodes correspond to partial partitions. That is, at leaf nodes, $\mathcal{F}$ defines a total function $[N] \rightarrow [k]$ whereas, at internal nodes, it defines a partial function with some initial partitions  not  having been assigned to a final bin. The goal of the algorithm is to discover a leaf node that maximizes the objective. As standard with any branch-and-bound procedure, Algorithm~\ref{alg:branch-and-bound} computes an admissible (upper) bound on the objective value of any leaf that can be reached from a given internal node $n$ and \emph{only} expands $n$ if its upper bound exceeds that of the best incumbent so far (we use the objective value of the heuristic partition from Algorithm ~\ref{alg:branch-and-bound} as an initial upper bound). Thus, there are two key parts of the algorithm: 

\begin{itemize}[leftmargin=*]
\item {\bf Evaluating objective at leaf nodes:} At every leaf node $l$, the algorithm evaluates the objective value from Definition~\ref{def:decomposed-query-selection}, and updates the incumbent $(\mathcal{F}^*, O_{\mathrm{best}})$ if there is an improvement (Lines~7--10 of Algorithm~\ref{alg:branch-and-bound}).
\item {\bf Upper bounding objective:} Algorithm~\ref{alg:branch-and-bound} invokes a helper procedure called {\sc ComputeBranchUB} to upper bound the objective that can be plausibly attained from any internal node. 
\end{itemize} 

In the remainder of the section, we explain the {\sc ComputeBranchUB} procedure in more detail. 

\subsection{Computing Upper Bounds for Internal Nodes}

For a partial partition $\mathcal{F}_p$, we upper bound the best
objective attainable by any completion using \textsc{ComputeBranchUB} (Algorithm~\ref{alg:compute-branch-ub}).
In particular, we bound the disambiguation power that \emph{any}
completion can achieve:
\[
\Bigl[\,1-  \max_{j\in[k]}\frac{|P_j(\mathcal{F}_p)|}{|\hs|}\Bigr].
\]
Recall from Definition ~\ref{def:induced-query} that $P_j(\mathcal{F}_p)$ denotes the subset of $\hs$ assigned to a cluster $C'_j$ induced by $\mathcal{F}_p$ . This upper bound is optimistic in that it treats all unassigned clusters as if they could be
perfectly separated from one another and from already assigned clusters (so disambiguation
is limited only by current class imbalance $\max_j |P_j(\mathcal{F}_p)|/|\hs|$).
If the bound does not exceed the incumbent
$O_{\mathrm{best}}$, the subtree is pruned; otherwise, we extend $\mathcal{F}_p$ by placing
an unassigned cluster into one of the bins and continue the search.

\begin{algorithm}[t]
\caption{\textsc{ComputeBranchUB}}
\label{alg:compute-branch-ub}
\begin{algorithmic}[1]
  \Require Partial partition $\mathcal{F}_p$ for clusters $1\dots p\!-\!1$
  \Ensure Upper bound $U$ on objective value for any completion of  $\mathcal{F}_p$
  \Statex \Return \quad
 $
    \Bigl[\,1- \max_{j\in[k]}\tfrac{|P_j(\mathcal{F}_p)|}{|\hs|}]
$  
  \Statex \textbf{where}
  \[
  \begin{aligned}
   P_j &:= \{\, h\in\hs \mid \exists\, i\in \mathrm{dom}(\mathcal{F}_p):  \mathcal{F}_p(i)=j \land \ h\in C_i \,\}\\
  \end{aligned}
  \]
\end{algorithmic}
\end{algorithm}

\begin{restatable}[Upper-Bound Soundness]{theorem}{UBSoundness}\label{thm:ub-sound}
Let $\phi$ be a precondition, $\mathcal{F}_p$ any partial partition, and
$U=\textsc{ComputeBranchUB}(\mathcal{F}_p)$.
For any completion $\mathcal{F}$ of $\mathcal{F}_p$, let $Q=(\phi, \psi_1, \ldots, \psi_k)$ be the induced query of $(\mathcal{F},\phi)$.
Then, we have  $\dpower(Q) - \lambda_{{\mathrm{post}}} \cdot \sum_{i=1}^k \complexity(\psi_i) \leq U$.
\end{restatable}

\begin{restatable}[Correctness of \textsc{MergeClusters}]{theorem}{MergeClusters}
\label{cor:merge-clusters}
\textsc{MergeClusters} solves Eq.~(\ref{eq:postobjective}) in Definition~\ref{def:decomposed-query-selection}; i.e., letting $\psi_1^*=\psi_1(\mathcal{F}),...,\psi_k^*=\psi_k(\mathcal{F})$ be the separators for the final partition mapping $\mathcal{F}$, then
\begin{align*}
(\postcond_1^*, \dots, \postcond_k^*) \in \argmax_{\substack{(\postcond_1, \dots, \postcond_k) \in \textsc{Cube}(\Upost)^k}}
\left[
\dpower(\precond, \postcond_1, \dots, \postcond_k) - \lambda_{\mathrm{post}} \cdot \sum_{i=1}^k \complexity(\postcond_i)
\right].
\end{align*}
\end{restatable}

\section{Common User Errors During Active Learning}\label{sec:errors}

Section ~\ref{sec:user-study} describes a user study wherein participants answer both the multiple-choice questions issued by \toolname, and the input-output questions issued by active learning techniques from prior work. In this section, we describe common user errors when answering queries in each of our four application domains. We focus our qualitative analysis on I/O responses for two reasons. First, MC errors are inherently opaque (i.e., an incorrect selection provides no observable trace of the participant’s reasoning), making it difficult to infer why a particular choice was made. Second, accuracy on MC questions was high, leaving too few incorrect responses to reveal consistent patterns of misunderstanding. In contrast, I/O responses offer richer insight into user behavior, since the outputs participants constructed expose specific reasoning errors. 

\begin{figure}[t]
\centering
\begin{tcolorbox}[colback=gray!2,colframe=black!10,boxsep=1ex,arc=2mm,width=\linewidth]
\footnotesize

\begin{subfigure}[t]{\linewidth}
\caption{Task description.}  
\label{fig:table-example-a}
Transform a table of individual messages into a summary over unique addresses.
Each input row represents one message with a sender and a recipient.
The output lists all unique addresses, counting how many messages each sent and received.
\end{subfigure}

\vspace{0.8ex}

\begin{subfigure}[t]{0.31\linewidth}
\caption{Input table}  
\label{fig:table-example-b}
\centering
\vspace{0.4ex}
\begin{tabular}{@{}c c c@{}}
\toprule
\textsf{id} & \textsf{sender} & \textsf{recipient} \\
\midrule
5 & A.1 & A.1 \\
1 & A.2 & B.2 \\
10 & A.2 & B.1 \\
8  & A.2 & B.2 \\
\bottomrule
\end{tabular}
\end{subfigure}
\hfill
\begin{subfigure}[t]{0.31\linewidth}
\caption{Correct output}  
\label{fig:table-example-c}
\centering
\vspace{0.4ex}
\begin{tabular}{@{}c c c@{}}
\toprule
\textsf{address} & \textsf{sender} & \textsf{recipient} \\
\midrule
A.1 & 1 & 1 \\
A.2 & 0 & 3 \\
B.1 & 1 & 0 \\
B.2 & 2 & 0 \\
\bottomrule
\end{tabular}
\end{subfigure}
\hfill
\begin{subfigure}[t]{0.31\linewidth}
\caption{Incorrect user output}  
\label{fig:table-example-d}
\centering
\vspace{0.4ex}
\begin{tabular}{@{}c c c@{}}
\toprule
\textsf{address} & \textsf{sender} & \textsf{recipient} \\
\midrule
A.1 & 1 & 1 \\
A.2 & 0 & \redcolor{2} \\
B.1 & 1 & 0 \\
B.2 & 2 & 0 \\
\bottomrule
\end{tabular}
\end{subfigure}

\end{tcolorbox}
\caption{An overview of a task in the \tables domain from our user study, including (a) an excerpt of the task description, (b) an input table presented to the user (c) the correct output table, and (d) an erroneous output table written by a participant.}
\label{fig:tables-example}
\end{figure}

\begin{figure}[t]
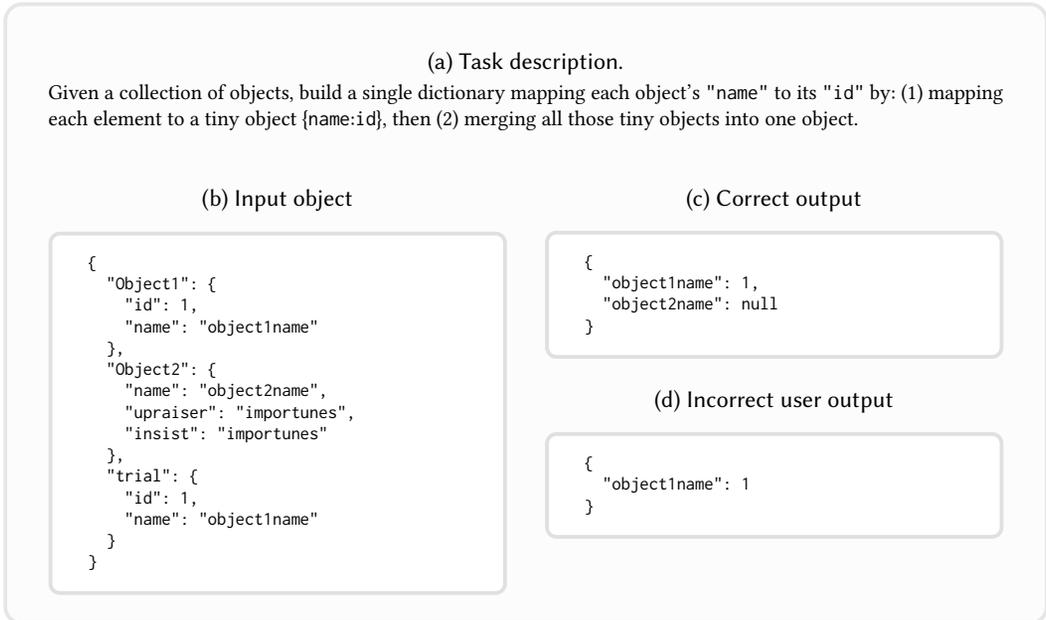

\centering
\begin{tcolorbox}[colback=gray!2,colframe=black!10,boxsep=1ex,arc=2mm,width=\linewidth]
\footnotesize

\begin{subfigure}[t]{\linewidth}
\caption{Task description.}  
\label{fig:json-example-a}                  
\vspace{0.3ex}

Given a collection of objects, build a single dictionary mapping each object's \texttt{"name"} to its \texttt{"id"} by:
(1) mapping each element to a tiny object \{\texttt{name}:\texttt{id}\}, then
(2) merging all those tiny objects into one object.
\end{subfigure}
\vspace{1.0ex}

\begin{subfigure}[t]{0.48\linewidth}
\caption{Input object}  
\label{fig:json-example-b}   
\centering
\vspace{0.4ex}

\begin{tcolorbox}[colback=white,colframe=black!12,boxsep=0.6ex,arc=1mm]
\scriptsize\ttfamily
\begin{verbatim}
{
  "Object1": {
    "id": 1,
    "name": "object1name"
  },
  "Object2": {
    "name": "object2name",
    "upraiser": "importunes",
    "insist": "importunes"
  },
  "trial": {
    "id": 1,
    "name": "object1name"
  }
}
\end{verbatim}
\end{tcolorbox}
\end{subfigure}
\hfill
\begin{subfigure}[t]{0.48\linewidth}
\centering

\begin{subfigure}[t]{\linewidth}
\caption{Correct output}
\label{fig:json-example-c}
\vspace{0.3ex}
\begin{tcolorbox}[colback=white,colframe=black!12,boxsep=0.6ex,arc=1mm]
\scriptsize\ttfamily
\begin{verbatim}
{
  "object1name": 1,
  "object2name": null
}
\end{verbatim}
\end{tcolorbox}
\end{subfigure}

\vspace{1ex}

\begin{subfigure}[t]{\linewidth}
\caption{Incorrect user output}
\label{fig:json-example-d}
\vspace{0.3ex}
\begin{tcolorbox}[colback=white,colframe=black!12,boxsep=0.6ex,arc=1mm]
\scriptsize\ttfamily
\begin{verbatim}
{
  "object1name": 1
}
\end{verbatim}
\end{tcolorbox}
\end{subfigure}

\end{subfigure}

\end{tcolorbox}

\caption{An overview of a task in the \trees domain from our user study, including (a) an excerpt of the task description, (b) an input JSON object presented to the user, (c) the correct output, and (d) an erroneous output written by a participant.}
\label{fig:jq-task}
\end{figure}

\begin{figure}[t]
\centering
\begin{tcolorbox}[colback=gray!2,colframe=black!10,boxsep=1ex,arc=2mm,width=\linewidth]
\footnotesize

\begin{subfigure}[t]{\linewidth}
\centering
\caption{Task description.}  
\vspace{.1cm}
\label{fig:image-task-a}
Find all images with bicycles below people.
\end{subfigure}

\vspace{0.8ex}

\begin{subfigure}[t]{0.48\linewidth}
\centering

\begin{subfigure}[t]{\linewidth}
\centering
\caption{Input image}
\label{fig:imageedit-input}
\includegraphics[width=0.95\linewidth]{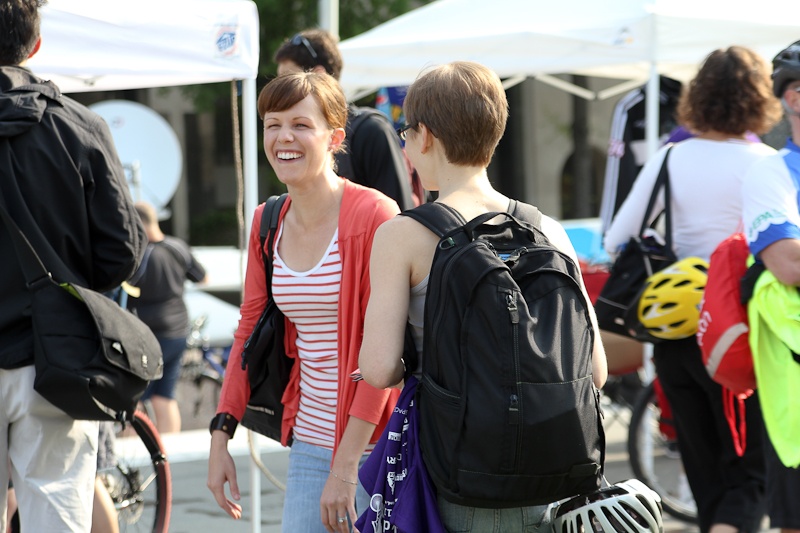}
\end{subfigure}

\vspace{1.2ex}

\begin{subfigure}[t]{\linewidth}
\centering
\caption{Annotated image}
\label{fig:imageedit-annotated}
\includegraphics[width=0.95\linewidth]
{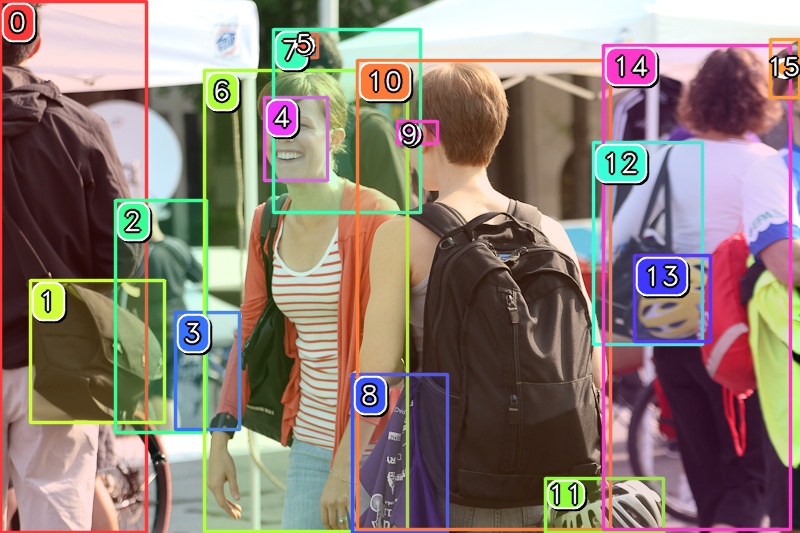}
\end{subfigure}

\end{subfigure}
\hfill
\begin{subfigure}[t]{0.48\linewidth}
\centering

\begin{subfigure}[t]{\linewidth}
\centering
\caption{Image legend}
\label{fig:image-task-legend}
\vspace{0.4ex}
\begin{tcolorbox}[colback=white,colframe=black!12,boxsep=0.6ex,arc=1mm,left=0.8ex,right=0.8ex]
\scriptsize
\begin{tabular}{rlrl}
0 & Person & 8 & Handbag \\
1 & Handbag & 9 & Glasses \\
2 & Person  & 10 & Person \\
3 & Bicycle & 11 & Helmet \\
4 & Face (eyes closed, smiling) & 12 & Handbag \\
5 & Glasses & 13 & Helmet \\
6 & Person  & 14 & Person \\
7 & Person  & 15 & Helmet \\
\end{tabular}
\end{tcolorbox}
\end{subfigure}

\vspace{1.2ex}

\begin{subfigure}[t]{\linewidth}
\centering
\caption{Correct output}
\label{fig:image-task-correct}
\vspace{0.3ex}
\begin{tcolorbox}[colback=white,colframe=black!12,boxsep=0.5ex,arc=1mm]
\scriptsize
\texttt{Selected objects: \{3\}}
\end{tcolorbox}
\end{subfigure}

\vspace{1.2ex}

\begin{subfigure}[t]{\linewidth}
\centering
\caption{Incorrect user output}
\label{fig:image-task-incorrect}
\vspace{0.3ex}
\begin{tcolorbox}[colback=white,colframe=black!12,boxsep=0.5ex,arc=1mm]
\scriptsize
\texttt{Selected objects: \{\}}
\end{tcolorbox}
\end{subfigure}

\end{subfigure}

\end{tcolorbox}
\caption{An overview of a task in the \imagesearch domain from our user study, including (a) an excerpt of the task description, (b) an input image, (c) an annotated version of the image with detected objects, (d) a legend of objects in the image, (e) the correct set of selected objects, (f) an erroneous output written by a participant.}
\label{fig:image-task}
\vspace{-0.15in}
\end{figure}

\subsection{\tables}

When answering input-output questions in this domain, the most common user error was mistyping or miscalculating a value in the output table. For instance, given the task described in Figure ~\ref{fig:table-example-a} and the input table in Figure ~\ref{fig:table-example-b}, users were supposed to write the output table in Figure ~\ref{fig:table-example-c}. However, one user apparently miscalculated the number of messages with recipient A.2, and instead wrote the table in Figure ~\ref{fig:table-example-d}. 


\subsection{\trees}

Answering input–output questions in this domain is particularly challenging because users must reason about structural transformations over nested JSON data. Even minor variations in how values are filtered, merged, or defaulted can silently alter the resulting structure. In the example shown in Figure~\ref{fig:json-example-a}, the ground-truth JQ program assigns \texttt{null} to objects missing an \texttt{id} field (Figure~\ref{fig:json-example-c}), whereas the user’s output (Figure~\ref{fig:json-example-d}) incorrectly omits these entries. This error illustrates how subtle aspects of JQ’s semantics can easily lead to incomplete or structurally inconsistent results. 

\subsection{\imageedit and \imagesearch}
In these domains, we used an interface similar to one presented in prior work on image search ~\cite{photoscout}. In particular, an input-output question presented the user with (1) an image, (2) a version of that image with objects highlighted and tagged, and (3) a legend that listed the labels and attributes of all objects. The user had to select the objects in the input image that corresponded to their task. For instance, given the task in Figure ~\ref{fig:image-task-a}, the user was presented with the images and legend in Figures ~\ref{fig:imageedit-input}, ~\ref{fig:imageedit-annotated}, and ~\ref{fig:image-task-legend}, respectively. This example highlights three difficult aspects of responding to input-output questions in these domains: 
\begin{enumerate}[leftmargin=*]
\item \textbf{Mispredictions of underlying neural networks.} Under the hood, \imageedit and \imagesearch rely on pretrained perception models to detect and label objects. These models occasionally misclassify or miss objects entirely. In this case, the object detector failed to identify the bicycles on the left and right, so these objects are not tagged in Figure ~\ref{fig:imageedit-annotated}. Such mispredictions add noise and uncertainty to the task. 
\item \textbf{Small or partially obscured objects.} While the bicycle with id 3 \emph{is} correctly identified by the system, it is in the background, partially occluded by several other objects, and slightly out of focus. Even with a legend, users often fail to identify such objects.  
\item \textbf{Reasoning about positional attributes.} This task  requires reasoning about the relative positions of bicycles and people in the image. A user may be uncertain about whether the person with id 2 is above the bicycle with id 3. 
\end{enumerate}
In this case, the ground-truth program selects object 3, since it is a bicycle that is below the person with id 2. When presented with this question, one user did not select any objects in this image; another erroneously selected objects 0 and 14. 

\section{Proofs of Theorems}\label{sec:proofs}

\subsection{Proof of Theorem~\ref{thm:optimality}}

\Optimality*

\begin{proof}
    Let $(\precond^*, \postcond_1^*, \ldots, \postcond_k^*)$ be a query generated by Algorithm ~\ref{alg:main}. 
    By Corollary~\ref{cor:sep-valid}, this query is valid. 
    By Theorem ~\ref{cor:omt-equivalence}, 
\[
\precond^* \in \argmax_{\psi \in \textsc{Cube}(\Upre)}
\bigl[\dppre(\psi) - \lambda_{\mathrm{pre}}\cdot \complexity(\psi)\bigr],
\]
and, by Theorem ~\ref{cor:merge-clusters},
\[
(\postcond_1^*, \dots, \postcond_k^*) \in \argmax_{\substack{(\postcond_1, \dots, \postcond_k) \in \textsc{Cube}(\Upost)^k}}
\left[
\dpower(\precond^*, \postcond_1, \dots, \postcond_k) - \lambda_{\mathrm{post}} \cdot \sum_{i=1}^k \complexity(\postcond_i)
\right].
\]
Thus, the query solves the Decomposed Query Selection Problem from Definition~\ref{def:decomposed-query-selection}.
\end{proof}

\subsection{Proof of Theorem~\ref{thm:main}}

\Main*

\begin{proof}
We make the following assumptions:
\begin{enumerate}
    \item The hypothesis space $\hs$ is finite and contains a program $P$ matching the user's intent.
    \item The user answers all queries correctly (i.e., in accordance with the semantics of $P$).
    \item \textsc{GenerateQuery} always outputs a query with at least 2 postconditions.
\end{enumerate}
By Lemma ~\ref{lemma:termination}, our active learning procedure terminates, and by Lemma ~\ref{lemma:correct-prog}, the user's intended program $P$ must be in the hypothesis space when termination occurs. Since all programs in the final hypothesis space are semantically equivalent, we must return a program $P^*$ that is equivalent to $P$.
\end{proof}

\begin{lemma}\label{lemma:correct-prog}
The intended program $P$ is never pruned from $\hs$.
\end{lemma}

\begin{proof}
We argue by induction on the iterations of Algorithm~\ref{alg:main}, maintaining the invariant
$I_t:\ P\in\mathcal H_t$ at the start of each loop.

\emph{Base case.} Initially, $P\in\mathcal H_0$ by construction of the hypothesis space.

\emph{Inductive case.} Assume $P\in\mathcal H_t$. Consider the update to $\mathcal H_{t+1}$.

\smallskip
\textbf{Case 1:} No query is issued (i.e., {\sc GetBestPrecondition} returns $\bot$ or {\sc GenerateQuery} fails).  
Line~\ref{alg:main:precondition-fail} refines $\Upre$ but does not modify $\mathcal H$. Hence $\mathcal H_{t+1}=\mathcal H_t$, so $P\in\mathcal H_{t+1}$.

\smallskip
\textbf{Case 2:} A query $Q=(\phi,\psi_1,\ldots,\psi_k)$ is issued and the user selects option $i$. Since the user answers questions correctly, $\{\phi\}\,P\,\{\psi_i\}$ holds, and hence we have $P\in\mathcal H_{t+1}$.
\end{proof}

\begin{lemma}\label{lemma:termination}
Algorithm~\ref{alg:main} always terminates.
\end{lemma}

\begin{proof}
We argue by induction on the cardinality of the hypothesis space $|\hs|$. 

\emph{Base case}. Suppose $|\hs| = 1$. Then trivially $\textsf{NumUnique}(\hs) = 1$, so we terminate. 

\emph{Inductive case}. Assume this lemma holds for $|\hs| \leq n$. We will show that the lemma holds for $|\hs| = n+1$.

\smallskip
\textbf{Case 1:} Suppose all programs in $\hs$ are semantically equivalent. Then $\textsf{NumUnique}(\hs) = 1$, and so active learning terminates.

\textbf{Case 2:} Suppose that there are semantically distinct programs in $\hs$. Then $\textsc{GetBestPrecondition}$ will compute a precondition $\phi$ that distinguishes at least one pair of programs in the hypothesis space. 
We will then generate a query $(\phi, \psi_1, \ldots, \psi_k)$ such that the postconditions $\{\psi_1, \ldots, \psi_k\}$ satisfy the coverage and mutual exclusion conditions of Definition ~\ref{def:valid-query} (by Corollary ~\ref{cor:sep-valid}), and each postcondition corresponds to a bin with at least one program (by construction). Since the user correctly selects the postcondition $\psi_i$ that corresponds to the semantics of $P$, and since $k \geq 2$, at least one program must be pruned. 
Thus, $|\hs| \leq n$ in the next iteration, and active learning will terminate by inductive hypothesis.
\end{proof}

\subsection{Proof of Theorem~\ref{cor:omt-equivalence}} 

\GetBestPrecond*

\medskip
\noindent\textbf{Assumptions and setup.}
By construction of {\sc GetDistinguishing}, for every pair $(i,j)$ the list
$\Phi(i,j)=\langle \Phi_{ij}^1,\dots,\Phi_{ij}^{K_{ij}}\rangle$ satisfies the following requirements:
(i) each $\Phi_{ij}^k$ is a \emph{cube over $\Upre$};
(ii) all $\Phi_{ij}^k$ are written over the same input variables;
(iii) each $\Phi_{ij}^k$ is a \emph{valid distinguisher} for $(P_i,P_j)$; and
(iv) if a cube $\psi\in\textsc{Cube}(\Upre)$ distinguishes $P_i$ and $P_j$,
then $\psi \Rightarrow \Phi_{ij}^k$ and $\mathrm{atoms}(\Phi_{ij}^k)\subseteq \mathrm{atoms}(\psi)$ for some $k$ by the definition of prime implicants(\emph{coverage–completeness}).


\begin{lemma}
\label{lem:inv}
Let $M$ satisfy (DP), (IA), (SAT). Then, if $p_{ij}(M)=1$ then $\precond_M$ distinguishes $P_i$ and $P_j$.
Consequently, $\sum_{i<j} p_{ij}(M) \le \dppre(\precond_M)$.
\end{lemma}
\begin{proof} By (SAT), $\precond_M$ is satisfiable and such that $\dppre(\precond_M)$ is well defined. 
From (DP), $p_{ij}(M)=1$ implies $d_{ij}^k(M)=1$ for some $k$.
By (IA), every atom occurring in the cube $\Phi_{ij}^k$ is selected among the $a_t$,
so $\precond_M \Rightarrow \Phi_{ij}^k$.
By construction of {\sc GetDistinguishing}, $\Phi_{ij}^k$ is a valid distinguisher;
hence $\precond_M$ distinguishes $P_i$ and $P_j$. Summing over pairs gives the desired inequality.
\end{proof}

\begin{lemma}
\label{lem:opt-counts}
Let $M^\star$ be an optimal satisfying assignment. Then
$\sum_{i<j} p_{ij}(M^\star)=\dppre(\precond_{M^\star})$.
\end{lemma}
\begin{proof}
Let $(i,j)$ be any pair that $\precond_{M^\star}$ distinguishes.
By coverage–completeness of $\Phi(i,j)$, there exists $k$ such that
\begin{equation}
\label{eq:containment}
\precond_{M^\star} \;\Rightarrow\; \Phi_{ij}^k
\quad\text{and}\quad
\mathrm{atoms}(\Phi_{ij}^k)\subseteq \mathrm{atoms}(\precond_{M^\star}).
\end{equation}
We claim $p_{ij}(M^\star)=1$. Suppose, for contradiction, that $p_{ij}(M^\star)=0$.
Construct $M'$ from $M^\star$ by setting $d_{ij}^k := 1$ and $p_{ij} := 1$,
leaving all $a_t$ (and all other $d$, $p$ and input variables) unchanged. We now argue that $M'$ is feasible and results in an objective improvement, thereby contradicting that $M^\star$ is a valid solution of the OMT problem.

\smallskip\noindent\emph{Feasibility of $M'$.}
\begin{itemize}[leftmargin=1.2em]
\item (IA): By \eqref{eq:containment}, every atom of $\Phi_{ij}^k$ already appears in $\precond_{M^\star}$,
so the corresponding $a_t$ are $1$ in $M^\star$ and remain $1$ in $M'$. Thus the implications required by (IA) for $d_{ij}^k$ hold.
\item (SAT): Should still holds because all $a_t$ selector variables and input variables are unchanged.
\item (DP): Setting $p_{ij}=1$ when some $d_{ij}^k=1$ preserves the biconditional for pair $(i,j)$;
all other pairs are unchanged, so (DP) holds globally.
\end{itemize}
\smallskip\noindent\emph{Objective improvement.}
No $a_t$ flips its assignment, so the complexity penalty is unchanged, while the coverage term increases by $1$.
Hence
\[
\sum_{i<j} p_{ij}(M') - \lambda_{\mathrm{pre}}\sum_t a_t(M')
\;>\;
\sum_{i<j} p_{ij}(M^\star) - \lambda_{\mathrm{pre}}\sum_t a_t(M^\star),
\]
contradicting optimality of $M^\star$.
Therefore $p_{ij}(M^\star)=1$ for every pair distinguished by $\precond_{M^\star}$.

\smallskip
\noindent\emph{From indicators to counts.}
Let
\[
S_p \;:=\; \{(i,j)\mid p_{ij}(M^\star)=1\}
\qquad\text{and}\qquad
S_d \;:=\; \{(i,j)\mid \precond_{M^\star}\ \text{distinguishes }(i,j)\}.
\]
We showed that, if $(i,j)\in S_d$, then $p_{ij}(M^\star)=1$, i.e.\ $S_d\subseteq S_p$.
By the first lemma, $p_{ij}=1$ implies that $(i,j)$ is distinguished, hence $S_p\subseteq S_d$.
Thus $S_p=S_d$, and taking cardinalities,
\[
\sum_{i<j} p_{ij}(M^\star) \;=\; |S_p| \;=\; |S_d| \;=\; \dppre(\precond_{M^\star}).
\]
Note that (SAT) ensures $\dppre(\precond_{M^\star})$ is well-defined as above.

\end{proof}

\begin{proposition}
\label{prop:obj-id}
For any optimal satisfying assignment $M^\star$,
\[
\sum_{i<j} p_{ij}(M^\star)\;-\;\lambda_{\mathrm{pre}}\sum_t a_t(M^\star)
\;=\;
\dppre(\precond_{M^\star})\;-\;\lambda_{\mathrm{pre}}\cdot\complexity(\precond_{M^\star}).
\]
\end{proposition}

\begin{proof}
Follows from Lemma~\ref{lem:opt-counts} and the assumption that $\complexity$ is defined in terms of the number of atomic predicates in the cube.
\end{proof}

\begin{theorem}[Equivalence to Objective~(1)]
\label{thm:eq-obj1}
Let $M^\star$ be a model satisfying the OMT encoding in Fig.~\ref{fig:omt-encoding}, and let $\precond^\star \in \textsc{Cube}(\Upre)$ be any proposition. If $\precond^\star = \precond_{M^\star}$, then
\[
M^\star \text{ is the optimal model for the encoding in Fig.}~\ref{fig:omt-encoding}
\iff
\precond^\star \in \argmax_{\psi \in \textsc{Cube}(\Upre)}
\Bigl[\dppre(\psi) - \lambda_{\mathrm{pre}}\cdot\complexity(\psi)\Bigr].
\]
Thus, optimizing problem in the OMT encoding in Fig.~\ref{fig:omt-encoding}
is consistent with our optimization goal described in Objective (1).
\end{theorem}

\begin{proof}
Assuming a fixed space of programs, let $\mathscr{M}$ denote the set of all models satisfying the hard constraints:
\[
\mathscr{M} := \{\, M \mid M \text{ satisfies the hard constraints (DP), (IA), and (SAT) of the OMT encoding in Fig.~\ref{fig:omt-encoding}} \,\}.
\]

Correspondingly, let
\[
\phi_{\mathscr{M}} := \{\, \precond_M \mid M \in \mathscr{M} \,\}
\]
denote the set of all preconditions induced by models in $\mathscr{M}$, where each
\[
\precond_M := \bigwedge_{\;a_t(M)=1} a_t
\]

For an arbitrary model $M^\star \in \mathscr{M}$ and $\precond^\star \in \textsc{Cube}(\Upre)$ 
such that $\phi_{M^\star} = \precond^\star$, we have the following as a result of Prop. ~\ref{prop:obj-id}:
\[
M^\star \text{ is optimal for the encoding in Fig.~\ref{fig:omt-encoding}}
\;\;\iff\;\;
\phi^\star \in 
\argmax_{\psi \in \phi_{\mathscr{M}}}
\Bigl[
\dppre(\psi) - \lambda_{\mathrm{pre}} \cdot \complexity(\psi)
\Bigr].
\]
Hence, it remains to show that
\[
\phi^\star \in 
\argmax_{\psi \in \phi_{\mathscr{M}}}
\Bigl[
\dppre(\psi) - \lambda_{\mathrm{pre}} \cdot \complexity(\psi)
\Bigr]
\quad\iff\quad
\phi^\star \in 
\argmax_{\psi \in \textsc{Cube}(\Upre)}
\Bigl[
\dppre(\psi) - \lambda_{\mathrm{pre}} \cdot \complexity(\psi)
\Bigr].
\]

\paragraph*{($\Leftarrow$) Direction.}
This direction is immediate, since by definition we have $\phi_{\mathscr{M}} \subseteq \textsc{Cube}(\Upre)$

\paragraph*{($\Rightarrow$) Direction.}
Let $\psi$ be an arbitrary cube in $\textsc{Cube}(\Upre)$.
If $\psi$ is unsatisfiable, then $\dppre(\psi)$ is undefined and $\psi$ cannot outperform $\precond^\star$. Otherwise, construct a model $M_\psi \in \mathscr{M}$ such that $\phi_{M_\psi} = \psi$ by assigning $a_t = 1$ exactly for every atom $A_t(x)$ appearing in $\psi$ and $a_t = 0$ otherwise. For each pair distinguished by $\psi$, coverage–completeness yields
$k$ with $\psi \Rightarrow \Phi_{ij}^k$ and $\text{atoms}(\Phi_{ij}^k) \subseteq \text{atoms}(\psi)$; set $d_{ij}^k=p_{ij}=1$. Set all remaining $d$/$p$ to $0$ and set all input variables according to any satisfying assignment of $\psi$. By this construction, all hard constraints (DP), (IA), and (SAT) are satisfied.
Specifically, (IA) holds because $\text{atoms}(\Phi_{ij}^k) \subseteq \text{atoms}(\psi)$, (DP) holds by direct construction of $p_{ij}$ and $d_{ij}^k$, and (SAT) holds since $\psi$ is satisfiable and we use one satisfying assignment for all input variables. Consequently, $\psi \in \phi_{\mathscr{M}}$.
Since $\phi^\star$ is optimal over $\phi_{\mathscr{M}}$, we have
\[
\dppre(\psi) - \lambda_{\mathrm{pre}}\cdot\complexity(\psi)
\;\le\;
\dppre(\phi^\star) - \lambda_{\mathrm{pre}}\cdot\complexity(\phi^\star).
\]
Since $\psi$ was chosen arbitrarily, it follows that $\precond^\star$ is optimal over all $\textsc{Cube}(\Upre)$.


\end{proof}

\subsection{Proof of Theorem~\ref{thm:sep}}

\ConstructSeparator*

\begin{proof}
Let $\mathcal{U}^+$ be the atom universe. Define
$\mathcal{A}\;=\;\{\,a\in\mathcal{U}^+\mid \mathrm{UNSAT}(\Phi^+\land \neg a)\,\}$.
We first show that restricting the search to $\mathcal{A}$ is valid, and then prove
termination, soundness, and optimality.

\paragraph{No loss by restricting to $\mathcal{A}$.}
Suppose $\hat\psi=\bigwedge_{a\in\hat S} a$ with $\hat S\subseteq\mathcal{U}^+$ satisfies
$\Phi^+\Rightarrow \hat\psi$. Then for each $a\in\hat S$ we have $\Phi^+\Rightarrow a$, which is
equivalent to $\mathrm{UNSAT}(\Phi^+\land \neg a)$; hence $a\in\mathcal{A}$ and
$\hat S\subseteq\mathcal{A}$. Thus any feasible cube over $\mathcal{U}^+$ is in fact a cube over
$\mathcal{A}$.
\paragraph{Termination.}
Since $\mathcal{A}$ is finite, there are finitely many cubes 
$\psi=\bigwedge_{a\in\mathcal{A}'}a$ with $\mathcal{A}'\subseteq\mathcal{A}$.
Each iteration of the algorithm either returns a valid separator or adds a 
blocking constraint derived from a counterexample model $m$ satisfying 
$\Phi^+\land\varphi_j\land\psi$. 
The blocking constraint ensures that any future candidate must differ from 
$\psi$ on at least one atom falsified by $m$, thereby eliminating $\psi$ 
and possibly others. 
Because there are finitely many possible cubes, the algorithm must 
eventually terminate.

\paragraph{Soundness.}
If the algorithm returns a formula $\psi$, it does so only when 
$\mathrm{UNSAT}(\varphi_j\land\psi)$ holds for all $j$.
Thus $\psi$ excludes all negative bins. 
Moreover, since all atoms in $\psi$ come from $\mathcal{A}$, each is implied 
by $\Phi^+$, and therefore $\Phi^+\Rightarrow\psi$. 
Hence $\psi$ satisfies the separation conditions:
\[
\Phi^+\Rightarrow\psi
\qquad\text{and}\qquad
\forall j.\ \mathrm{UNSAT}(\varphi_j\land\psi).
\]

\paragraph{Optimality.}
At each iteration, the MaxSAT solver returns the cube $\psi$ that 
minimizes the objective $\sum_{a\in\mathcal{A}}s_a$ among all cubes consistent 
with the accumulated blocking constraints. 
Each constraint added by a counterexample $m$ is necessary: any cube 
omitting all atoms falsified by $m$ would satisfy $\varphi_j$ 
for the same $j$, and thus cannot be feasible. 
Thus, no feasible separator is ever ruled out. 
When the algorithm terminates with a feasible $\psi$, it must therefore 
minimize the objective among all feasible cubes.
\end{proof}

\begin{corollary}[Separator Validity]\label{cor:sep-valid}
Let $\{B_i\}_{i=1}^k$ be the bins induced by a fixed precondition $\precond$, and suppose that for each $i$, 
$ 
\psi_i \;=\; \textsc{ConstructSeparator}\bigl(B_i, \{B_j\}_{j \neq i}, \Upost\bigr) \neq \bot.
$ 
Then the postconditions $\{\psi_1, \dots, \psi_k\}$ satisfy the coverage and mutual exclusion 
conditions of Definition~\ref{def:valid-query} with respect to the hypothesis space $\hs$.
\end{corollary}

\begin{proof} Let $\prog \in \hs$, and note $\prog \in B_i$ for some $i$.
    
\paragraph{Coverage.} 
By construction, $\textsf{sp}(\prog, \precond) \implies \psi_i$, and so $\vDash \{\precond\}\prog\{\psi_i\}$.

\paragraph{Mutual exclusion.} Consider another postcondition $\psi_j$ with $i \neq j$. Towards contradiction, suppose $\vDash \{\precond\}\prog\{\psi_j\}$. Then $\textsf{sp}(\prog, \precond) \implies \psi_j$, and so $\textsf{sp}(\prog, \precond) \wedge \psi_j$ is satisfiable. However, since $\prog \in B_i$, $\textsf{sp}(\prog, \precond) \implies \bigvee_{Q\in B_i} \textsf{sp}(Q, \precond)$. Hence, $\psi_j \wedge \bigvee_{Q\in B_i} \textsf{sp}(Q, \precond)$ must be satisfiable, which contradicts the $\text{UNSAT}(\psi_j, \bigvee_{Q\in B_i} \textsf{sp}(Q, \precond))$ property which we have by Theorem ~\ref{thm:sep}. $\Rightarrow\!\Leftarrow$ Therefore, we have $\neg(\vDash \{\precond\}\prog\{\psi_i\} \wedge \vDash \{\precond\}\prog\{\psi_j\})$ 
\end{proof}

\subsection{Proof of Theorem ~\ref{thm:ub-sound}}

\UBSoundness*

\begin{proof}
Let $P_j$ be the per-bin programs defined by
Algorithm~\ref{alg:compute-branch-ub} for the partial partition $\mathcal F_p$.
Fix any completion $\mathcal F \supseteq \mathcal F_p$, and let
$Q=(\phi,\psi^{\mathcal F}_1,\dots,\psi^{\mathcal F}_k)$ be the induced query. Note that completions only add hypotheses to bins, so
$\mu_j^{\mathcal F}\ge |P_j|/|\hs|$ for all $j$, hence
$\max_j \mu_j^{\mathcal F}\ \ge\ (\max_j |P_j|/|\hs|)$.
Then, 
\begin{equation}
\label{eq:ub-dp}
\dpower(Q)\ =\ \dpower(\mathcal F,\phi)\ =
1- \max_{j\in[k]}\tfrac{|P_j|}{|\hs|}.
\end{equation}


Therefore, 
\[
\dpower(Q)\;-\;\lambda_{\mathrm{post}}\sum_{j=1}^k \complexity(\psi^{\mathcal F}_j)
\ \le\
\Bigl[\,1- \max_{j}\tfrac{|P_j|}{|\hs|}\Bigr]
\;=\; U,
\]
\end{proof}

\subsection{Proof of Theorem ~\ref{cor:merge-clusters}}

\MergeClusters*

\begin{proof}
The procedure \textsc{BranchAndBound} is a standard branch-and-bound search over a finite set of admissible completions.
By Theorem~\ref{thm:ub-sound}, \textsc{ComputeBranchUB} provides a sound upper bound on the objective for any subtree, ensuring that no branch containing a better solution is pruned.
Since every completion is either evaluated directly or pruned only when its achievable objective is $\le O_{\mathrm{best}}$, the final incumbent partition $\mathcal{F}$ must achieve the global maximum of \textsc{EvaluateObjective}. Lemma~\ref{lem:mergecluster-completeness} shows that optimizing over the entire space of mappings $\mathcal{F} : [N] \to [k]$ is equivalent to optimizing over all possible postconditions in $\textsc{Cube}(\Upost)^k$.
Hence the returned partition~$\mathcal{F}$ yields a query $(\phi, \psi_1(\mathcal{F}), \dots, \psi_k(\mathcal{F}))$ that is valid (by Lemma~\ref{lem:mergecluster-query-valid}) and satisfies
\begin{align*}
(\phi,\,\psi_1(\mathcal{F}),\dots,\psi_k(\mathcal{F})) \in \argmax_{\substack{(\postcond_1, \dots, \postcond_k) \in \textsc{Cube}(\Upost)^k}}
\left[
\dpower(\precond, \postcond_1, \dots, \postcond_k) - \lambda_{\mathrm{post}} \cdot \sum_{i=1}^k \complexity(\postcond_i)
\right].
\end{align*}.
\end{proof}

\begin{lemma}[\textsc{MergeCluster} Query Validity]\label{lem:mergecluster-query-valid}

Fix a precondition $\precond$ and a hypothesis space $\hs$, and let $\{C_1, C_2, \ldots\}$ be the partition of $\hs$ induced by equality of strongest postconditions $sp(\precond, P)$.
For any mapping $\mathcal{F} : [N] \to [k]$ as in Definition~\ref{def:induced-query},
$(\precond, \psi_1(\mathcal{F}), \ldots, \psi_k(\mathcal{F}))$
is a valid query satisfying the coverage and mutual-exclusion conditions of Definition~\ref{def:valid-query}.
Moreover, for any $P \in C_i \subseteq \hs$, we have
\[
\mathsf{sp}(\precond, P) \Rightarrow \psi_j(\mathcal{F})
\iff \quad
j = \mathcal{F}(i).
\]
\end{lemma}

\begin{proof}
For any $P \in \hs$, since $\{C_1, C_2, \ldots\}$ forms a valid partition of $\hs$, 
there exists a unique index $i$ such that $P \in C_i$. 
Let $j = \mathcal{F}(i)$.
By Definition~\ref{def:induced-query}, each postcondition $\psi_j(\mathcal{F})$ is constructed as
\[
\psi_j(\mathcal{F})
:= \textsc{ConstructSeparator}\bigl(P_j(\mathcal{F}),\, \mathcal{N}_j(\mathcal{F}),\, \Upost\bigr).
\]
From this definition, we have $P \in P_j(\mathcal{F})$.  
Then, by Theorem~\ref{thm:sep} (behavior of \textsc{ConstructSeparator}), 
it follows that 
\[
\mathsf{sp}(\precond, P) \Rightarrow \psi_j(\mathcal{F}).
\]
Similarly, for any $k \neq j$, by construction we have $P \in \mathcal{N}_k(\mathcal{F})$. 
Applying Theorem~\ref{thm:sep} again, we obtain
\[
\psi_k(\mathcal{F}) \Rightarrow \neg\,\mathsf{sp}(\precond, P).
\]
Hence, $\mathsf{sp}(\precond, P)$ implies exactly one postcondition 
$\psi_j(\mathcal{F})$, where $j = \mathcal{F}(i)$, 
satisfying both coverage and mutual exclusion. 
\end{proof}

\begin{lemma}[\textsc{MergeCluster} Completeness]\label{lem:mergecluster-completeness}
Let $\{C_1, \ldots, C_N\}$ be a valid $\mathsf{sp}$-partition over $\hs$.
For any $(\psi'_1, \ldots, \psi'_k) \in \textsc{Cube}(\Upost)^k$ such that
$(\precond, \psi'_1, \ldots, \psi'_k)$ forms a valid query,
there exists a mapping $\mathcal{F} : [N] \to [k]$ for which
$(\precond, \psi_1(\mathcal{F}), \ldots, \psi_k(\mathcal{F}))$
is another valid query that is at least as optimal as
$(\precond, \psi'_1, \ldots, \psi'_k)$.
\end{lemma}

\begin{proof}

We will construct such $\mathcal{F}$ as follows

\paragraph{Formation of $\mathcal{F}$}
For each $i \in [N]$, take any program $P \in C_i$.  
By the coverage and mutual-exclusion condition, there is exactly one index $j \in [k]$  
such that $\mathsf{sp}(P, \precond) \implies \psi'_j$.  
We therefore define $\mathcal{F}(i) := j$.  
By definition, $\mathcal{F}$ satisfies:
\[
\mathcal{F}(i) = j
\iff
\forall P \in C_i,\;
\mathsf{sp}(P, \precond) \Rightarrow \psi'_j.
\]

We now analyze how the objective of $(\precond, \psi_1(\mathcal{F}), \ldots, \psi_k(\mathcal{F}))$ compares with that of $(\precond, \psi'_1, \ldots, \psi'_k)$.

\paragraph{Distinguishing Power.}

Consider any two distinct programs $P_1, P_2 \in \hs$, where $P_1 \in C_a$ and $P_2 \in C_b$ for some $a,b \in [N]$. 
If $(\precond, \psi_1(\mathcal{F}), \ldots, \psi_k(\mathcal{F}))$ fails to distinguish $P_1$ and $P_2$, then there exists some $k$ such that
\[
\mathsf{sp}(P_1, \precond) \lor \mathsf{sp}(P_2, \precond) \implies \psi_k(\mathcal{F}).
\]
By Proposition~\ref{lem:mergecluster-query-valid}, this is equivalent to $k = \mathcal{F}(a) = \mathcal{F}(b)$. 
Hence, by the definition of $\mathcal{F}$, we have
$\mathsf{sp}(P_1, \precond) \implies \psi'_k$
and
$\mathsf{sp}(P_2, \precond) \implies \psi'_k$,
so that $(\precond, \psi'_1, \ldots, \psi'_k)$ also fails to distinguish $P_1$ and $P_2$. 
Because the reasoning above is bidirectional (based on an if-and-only-if relationship), 
the two query sets exhibit identical distinguishing power:
\[
\dpower(\precond, \psi_1(\mathcal{F}), \ldots, \psi_k(\mathcal{F})) =
\dpower(\precond, \psi'_1, \ldots, \psi'_k).
\]

\paragraph{Query Complexity.}

Consider any $j \in [k]$.  
By Definition~\ref{def:induced-query}, we can simplify $P_j(\mathcal{F})$ as:
\[
P_j(\mathcal{F}) = \{\, h \in \hs \mid \exists i,~\mathcal{F}(i) = j \land h \in C_i \,\}
= \{\, h \in \hs \mid \mathsf{sp}(h, \precond) \implies \psi'_j \,\}.
\]
Hence $\left( \bigvee_{P\in P_j(\mathcal{F})}\mathsf{sp}(P,\precond) \right) \implies \psi'_j$.   
Similarly, we can express:
\[
N_j(\mathcal{F}) = \{\, h \in \hs \mid \exists r \neq j,~\mathsf{sp}(h, \precond) \implies \psi'_r \,\}.
\]
By the mutual-exclusion property of $(\precond, \psi'_1, \ldots, \psi'_k)$, this implies
\[
\mathrm{UNSAT}(\psi'_j \land \bigvee_{P\in N_j(\mathcal{F})}\mathsf{sp}(P,\precond) ).
\]
By Theorem~\ref{thm:sep} and the construction of $\psi_j(\mathcal{F})$, we know that
$\psi_j(\mathcal{F})$ is the \emph{smallest} cube in $\Upost$ satisfying both
\[
\left( \vee_{P\in P_j(\mathcal{F})}\mathsf{sp}(P,\precond) \right)  \implies \psi_j(\mathcal{F})
\quad\text{and}\quad
\mathrm{UNSAT}(\psi_j(\mathcal{F}) \land \vee_{P\in N_j(\mathcal{F})}\mathsf{sp}(P,\precond) ).
\]
We just showed that $\psi'_j$ also satisfies these two constraints; hence, it follows that
$\psi_j(\mathcal{F})$ cannot be more complex than $\psi'_j$.  
Therefore, for all $j \in [k]$,
\[
\complexity(\psi_j(\mathcal{F})) \leq \complexity(\psi'_j).
\]

\paragraph{Conclusion} We can now compare the objective and conclude
\[
\dpower(\precond, \postcond'_1, \dots, \postcond'_k) - \lambda_{\mathrm{post}} \cdot \sum_{i=1}^k \complexity(\postcond_i)
\leq
\dpower(\precond, \psi_1(\mathcal{F}), \ldots, \psi_k(\mathcal{F}))  - \lambda_{\mathrm{post}} \cdot \sum_{i=1}^k \complexity(\postcond_i).
\]

\end{proof}

\end{document}